\documentclass[twocolumn,superscriptaddress,english,showpacs,prb,longbibliography]{revtex4-2}
\usepackage[colorlinks=true,urlcolor=blue,citecolor=blue,linkcolor=blue]{hyperref} 

\usepackage{amsmath}
\usepackage{subcaption}
\usepackage{graphicx}
\usepackage{textcomp}
\usepackage{bm}
\usepackage{color}
\usepackage{amssymb}
\usepackage{graphicx}
\usepackage{color}
\usepackage{mathrsfs}
\usepackage{float}
\usepackage{indentfirst}
\usepackage{cleveref}
\usepackage{tikz}
\usepackage{balance}
\usepackage[linesnumbered, ruled, vlined]{algorithm2e}
\labelformat{algocf}{Algorithm.~#1}
\usepackage{physics}
\usepackage{tikz}
\usepackage{mathdots}
\usepackage{yhmath}
\usepackage{cancel}
\usepackage{siunitx}
\usepackage{array}
\usepackage{multirow}
\usepackage{gensymb}
\usepackage{tabularx}
\usepackage{extarrows}
\usepackage{booktabs}
\usepackage{amsthm}
\usetikzlibrary{fadings}
\usetikzlibrary{patterns}
\usetikzlibrary{shadows.blur}
\usetikzlibrary{shapes}
\newtheorem{theorem}{Theorem}
\newtheorem{definition}{Definition}

\newcommand{\SWAP}{{\rm SWAP}}
\newcommand{\SU}{{\rm SU}}

\newcommand{\CNOT}{{\rm CNOT}}

\newcommand{\elementgate}{ 
            \raisebox{-0.2em}{\begin{tikzpicture}[x=0.11pt,y=0.11pt,yscale=-1.7,xscale=1.7]
\draw  [line width=0.5]  (38.45,8) -- (62.3,31.85) -- (38.45,55.7) -- (14.6,31.85) -- cycle ;
\draw  [fill={rgb, 255:red, 0; green, 0; blue, 0 }  ,fill opacity=1 ] (26.6,25.79) .. controls (26.6,22.42) and (29.33,19.69) .. (32.7,19.69) .. controls (36.07,19.69) and (38.8,22.42) .. (38.8,25.79) .. controls (38.8,29.15) and (36.07,31.89) .. (32.7,31.89) .. controls (29.33,31.89) and (26.6,29.15) .. (26.6,25.79) -- cycle ;
\draw [line width=0.5]    (49.97,19.29) -- (60.26,9.29) ;
\draw [line width=0.5]    (15.97,53.86) -- (26.26,43.86) ;
\draw [line width=0.5]    (16.26,9) -- (26.83,19.57) ;
\draw [line width=0.5]    (51.69,43.57) -- (62.26,54.14) ;

            \end{tikzpicture}}
}

\begin{document}

\preprint{APS/123-QED}

\title{Universal quantum computing with a single arbitrary gate} 

\author{Zhong-Yi Ni}
 \email{zni573@connect.hkust-gz.edu.cn}
 \affiliation{Hong Kong University of Science and Technology (Guangzhou), Guangzhou, China}
\author{Yu-Sheng Zhao}
\email{yzhao053@connect.hkust-gz.edu.cn}
 \affiliation{Hong Kong University of Science and Technology (Guangzhou), Guangzhou, China}
\author{Jin-Guo Liu}
 \email{jinguoliu@hkust-gz.edu.cn}
 \affiliation{Hong Kong University of Science and Technology (Guangzhou), Guangzhou, China}

\begin{abstract}
    This study presents a roadmap towards utilizing a single arbitrary gate for universal quantum computing. 
    Since two decades ago, it has been widely accepted that almost any single arbitrary gate with qubit number $>2$ is universal. Utilizing a single arbitrary gate for compiling is beneficial for systems with limited degrees of freedom, e.g. the scattering based quantum computing schemes. However, how to efficiently compile the wanted gate with a single arbitrary gate, and finally achieve fault-tolerant quantum computing is unknown.
    In this work, we show almost any target gate can be compiled to precision $\epsilon$ with a circuit depth of approximately $\log(\epsilon^{-1})$ with an improved brute-force compiling method. Under the assumption of reasonable classical resource, we show the gate imperfection can be lowered to $10^{-3}$. By treating the imperfection as coherent error, we show that the error can be further reduced by roughly two orders of magnitude with a measurement-free quantum error correction method.
\end{abstract}

\maketitle


\section{Introduction}
Universal quantum computing is one of the ultimate goals of quantum technology.
It is believed that quantum computers can solve certain problems intractable for classical computers~\cite{shor1999polynomial,georgescu2014quantum}.
However, current quantum computing technology is still far from achieving this goal~\cite{Preskill2018}. Scalability and gate precision are the two major challenges.
Although a number of physical systems have been proposed for scalable quantum computing, their scalability is still limited.
The state-of-the-art quantum computing devices, such as neutral atoms array~\cite{Bluvstein2023} and superconducting devices~\cite{Arute2019} can reach a few hundreds of qubits. The further scaling is difficult due to both the power limitation of experimental devices, and the need for complex control.

Scalability is easily accessible in some natural compounds with discrete degrees of freedom, such as the sequence controlled polymers~\cite{Lutz2013}.
Millions of bits can be encoded in a DNA sequence~\cite{Doricchi2022}, providing a robust degrees of freedom to manipulate.
Instead of using external control, utilizing the natural Hamiltonian of these systems for universal quantum computing is possible. One of the most promising schemes is the scattering based quantum computing~\cite{Childs2009,Childs2013}.
However, the scattering Hamiltonian corresponds to the adjacency matrix of an unweighted graph. Implementing one is extremely challenging in a natural compounds, since interaction strengths are naturally non-uniform.
For a long time, the study to scattering based quantum computing is limited to the theoretical models, such as the Fermi Hubbard model~\cite{bao2015universal} and Frenkel exciton model~\cite{yurke2010passive,yurke2023implementation,sup2024ballistic}.

The natural compounds can implement quantum operations. However, these operations usually have a random form due to the limited degrees of freedom.
Can we still build a universal quantum computer with natural compounds?
The answer is yes, but a roadmap is missing.
Almost any single arbitrary quantum gate with qubit number $>2$ is universal in computational power, while the only exceptions have a zero measure in population~\cite{Deutsch1995, Lloyd1995}.
The missing piece is an algorithm to efficiently compile a quantum circuit with a single arbitrary gate.
Most of the existing approaches for quantum compiling~\cite{Dawson2005,booth2018comparing,booth2018comparing,venturelli2018compiling,paler2023machine,khatri2019quantum}
are based on the assumption that the compiling gate set is symmetric, i.e., the inverse of each gate is also available. This is an overly strong assumption for the arbitrary gate based quantum computing framework.
The latest progress of inverse-free Solovay-Kitaev algorithm~\cite{Bouland2021} provides a ${\rm depth} \sim \log^{12}(\epsilon^{-1})$ algorithm that works for inverse free gate sets.
However, the polynomial order, $12$, of this method is too large for practical implementation.

Furthermore, we want to achieve fault-tolerance.
A logical qubit operation fidelity of $10^{-10}$ is expected to make quantum computing practically useful~\cite{Gidney2021}.
Quantum error correction is required to achieve this level of fidelity. However, the traditional quantum error correction methods are not suitable for the scattering based quantum computing framework.
The measure and feed forward error correction~\cite{beale2018quantum,shor1995scheme,steane1996error} requires external control during the computation process, which is not available in scattering based quantum computation.
Developing a measurement-free~\cite{perlin2023fault,heussen2024measurement,veroni2024optimized} quantum error correction (MF-QEC) method is crucial.

In this letter, we partially solve the compilation and error correction issues in arbitrary gate based quantum computing and show a clear path towards achieving universal quantum computing with natural compounds.
Our discussion is based on scattering based quantum computing framework, however the results are general and can be applied to other quantum computing frameworks.
In \Cref{sec:background}, we introduce scattering based quantum computing as a background of our problem. In \Cref{sec:compiling}, we show that any two-qubit gate can be compiled with a circuit depth logarithmic to gate infidelity. In \Cref{sec:qec}, we show that the compilation error can be further reduced by roughly two orders of magnitude with a MF-QEC method.

\section{Background: Scattering based quantum computing}\label{sec:background}
The scattering based quantum computing mentioned in this work is also known as the weighted version of the quantum walk based quantum computing.
Its framework is shown in \Cref{fig:main} (a).
(1) We first excite some synchronized fermions or bosons with a specific momentum $k$ from the input side, and quantum information is encoded in some internal degrees of freedom of these particles. (2) These particles with certain momentum propagate through chains and interact with others at the scattering centers denoted by $G_1, G_2,\ldots ,G_5$. In the transmission limit, the scattering process is equivalent to applying a unitary gate on the input state. (3) Detectors are placed on the output side to measure the outcomes.
A scattering center $G$ in the circuit is further decomposed to a sequence of elementary scattering centers (\elementgate) as show in \Cref{fig:main} (b).
The elementary scattering center implements an arbitrary two-qubit gate.
By flipping the scattering center horizontally and/or vertically, as shown in \Cref{fig:main} (c), four different variants of the gate can be generated.
By tuning the sequence of variants, any two-qubit gate can be approximated efficiently.
In the rest of this section, we will discuss the details of this framework.

\begin{figure}[h]
    \centering
    \includegraphics[width=0.9\columnwidth,trim={0 0 0 5cm},clip]{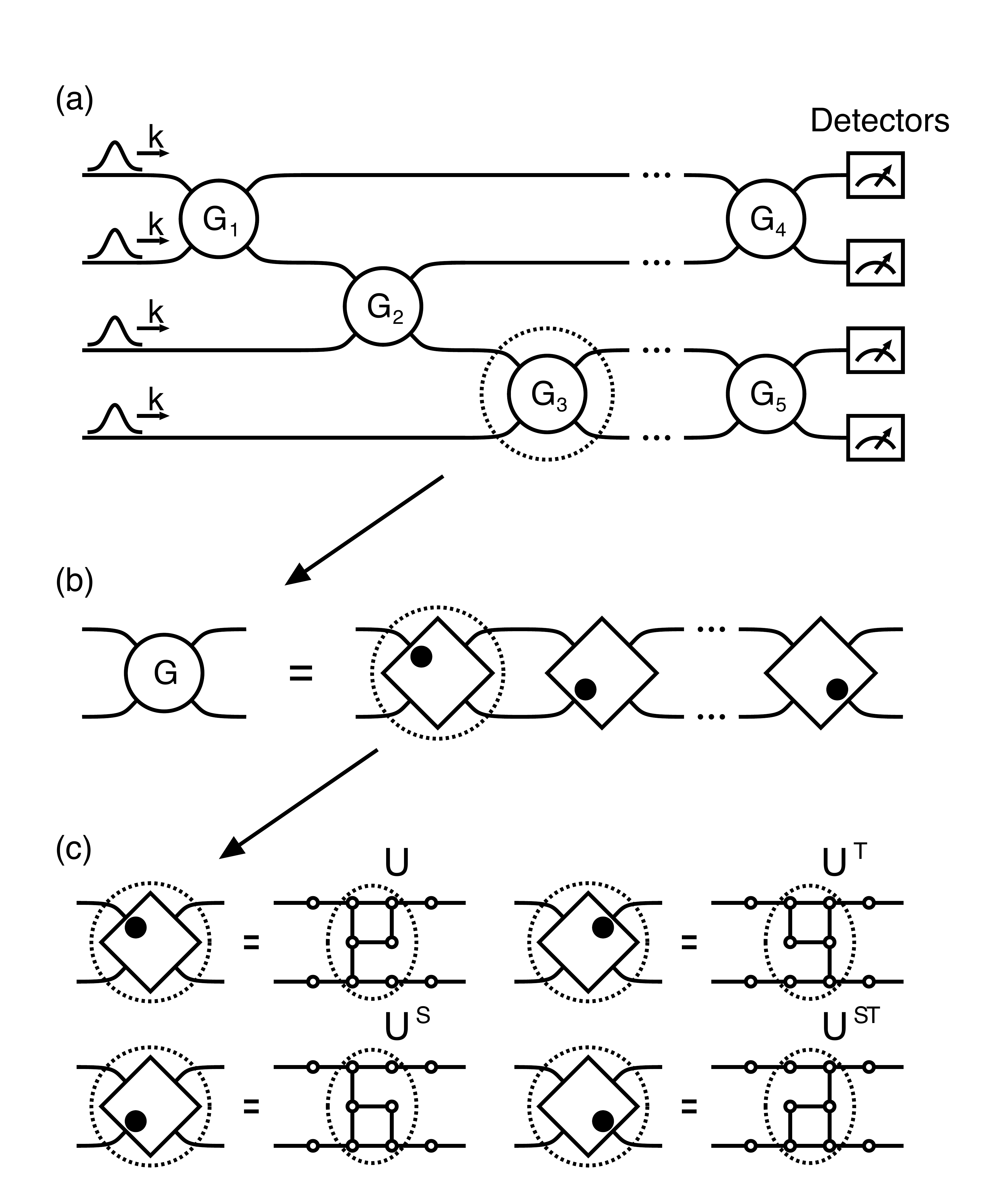}
    \caption{Scattering based quantum computing scheme. (a) The computation starts with exciting some (quasi-)particles with a specific momentum $k$ from the input channel. Those particles will propagate through the chains and interact with each other at the scattering centers denoted by $G_1,G_2, \ldots, G_5$. This scattering process is equivalent to applying a unitary on the input state. Detectors are placed on the end of output channels to measure the outcomes. (b) Each quantum gate is compiled with an elementary scattering center and its variants. (c) Given the elementary scattering center that implements a quantum gate $U$, three other variants can be generated by reflecting it horizontally or vertically.}\label{fig:main}
\end{figure}

\subsection{Implement Quantum Gates with Scattering Centers}
The scattering process of quasi-particles propagating through chains can be described by a scattering matrix. Given a system with $n$ semi-infinite chains (or channels) attached to a scattering center, the scattering matrix is denotes as $S \in \mathbb{C}^{n\times n}$, where $S_{i,j}$ is the probability amplitude of a particle scattering from channel $j$ into channel $i$.
Channels are grouped into inputs and outputs, then the scattering matrix has the following block-matrix form~\cite{texier2001scattering, varbanov2009quantum},
\begin{equation}
    S = \begin{pmatrix}
        S_{in,in} & S_{in,out} \\
        S_{out,in} & S_{out,out}
    \end{pmatrix}.
\end{equation}
Given a plane wave, $\psi_\alpha(x) = e^{ikx}|\alpha\rangle$ with some internal degree of freedom $\alpha$, entering from the input channel with momentum $k$, the wave function after scattering is a superposition of $e^{-ikx}S_{in,in}|\alpha\rangle$ in the input channel and $e^{ikx}S_{out,in}|\alpha\rangle$ in the output channel.
Running a circuit in scattering based quantum computing consists of multiple scattering processes. The scattering centers must be transparent, i.e., $S_{in,in} = S_{out,out} = 0$. Otherwise, wave amplitude may be damped every time the quasi-particles are scattered. 
By the conservation of probability, the scattering matrix is a unitary matrix, hence $S_{in,out}$ and $S_{out,in}$ are also unitary matrices,
i.e. the wave function in the input channels is scattered to the output channels through a unitary transformation.
Furthermore, a time reversal symmetric scattering process implies 
$S_{in,out} = S_{out,in}^T$~\cite{varbanov2009quantum}.
In the original scattering based scheme, the scattering centers for Hadamard gate, phase gate, basis changing gate, and $\sqrt{CZ}$ gate are designed, which are universal~\cite{Childs2009,Childs2013}.
By integrating the scattering centers into a quantum circuit, any quantum circuit can be implemented.

In the following, we consider a more general setup where the entries in a Hamiltonian can be random numbers and the Hamiltonian can not be mapped to the Laplacian of an unweighted graph. Therefore, 
it is unlikely to implement a specific quantum gate precisely.
We make the following presumptions:
(1). The gate implementable is a random one, rather than any specific one.
    However, the gate available can be very accurate such that the random error is negligible.
(2). Only one gate is available. We do not expect to find two structures that are simultaneously transparent for the same momentum.

\section{Efficient Compilation with a single arbitrary quantum gate}\label{sec:compiling}
Given the above restriction, we discuss how to compile a quantum circuit. Chen. et. al~\cite{chen2024one} show that a quantum circuit can be much more efficiently compiled with $\SU(4)$ gates than traditional finite gate set. In the following, we show that $\SU(4)$ gates can be efficiently compiled by stacking an arbitrary gate and its variants. 
\begin{definition}
    \label{definition:compilation-n-qubit}
    $\SU(N)$ Gate Compilation Problem
    
    Input: Given a set of $N\times N$ unitary matrices
    $\mathcal{U}=\{U_1,U_2,\dots,U_M\}$, a target unitary $\Gamma \in \SU(N)$, a depth $d$ and a
    precision $\epsilon>0$.
    
    Output:  A sequence of $d$ unitary matrices
    $U_{i_1},U_{i_2},\dots,U_{i_d}$ such that
    $\|U_{i_d}U_{i_{d-1}}\dots U_{i_1}- \Gamma\| <\epsilon$ if it exists and FALSE otherwise.
\end{definition}
$\|\cdot\|$ is the Frobenius norm and $\mathcal{U}$ is often assumed to be a universal gate set, which means any unitary matrix can be approximated by a sequence of unitary matrices in $\mathcal{U}$.
Finding a universal gate set turns out to be not as hard as it first appeared.
In the scattering based computing scheme, a single arbitrary $n$-qubit quantum gate can induce four different unitary matrices by flipping the scattering center horizontally and/or vertically as illustrated in \Cref{fig:main} (c).
By flipping the scattering center horizontally, we obtain $U^T$ under time reversal, which is the transpose of $U$.
By flipping the scattering center vertically, we obtain $U^S \equiv \SWAP \cdot U \cdot \SWAP$.
Hence, the four induced unitaries are
\begin{equation}\label{eq:four-unitaries}
    \mathcal{U} = \{U, U^S, U^{T}, U^{ST}\}
\end{equation}
and subset of which $\{U, U^S\}$ can already universally generate an arbitrary target gate~\cite{Deutsch1995,Lloyd1995} given the $U$ does not belong to the set of special unitaries that has a measure zero in the $\SU(N)$ manifold.
\paragraph{Why existing compiling methods are not good enough? --} 
Universality does not indicate that the compiling problem is easy~\cite{soeken2020boolean,botea2018complexity,ferrari2021deterministic}.
Deutsch et. al.~\cite{Deutsch1995} showed how to achieve a target gate precision $\epsilon_{F}$ with circuit depth $d\sim 1/\epsilon_{F}$ in the above setup. Here $\epsilon_F$ is the Frobenius norm of the difference between the target gate and the compiled gate.
The linearly increasing circuit depth with respect to $1/\epsilon_{F}$ is not acceptable in practice.
For a long time, the quantum computing community has been seeking for a universal gate set such as $\{H, T, \CNOT\}$ for better compilation of $\SU(4)$, where $H$ is the Hadamard gate, $T$ is the $\pi/8$ gate, and $\CNOT$ is the controlled NOT gate.
The most famous work is the Solovay-Kitaev algorithm~\cite{Dawson2005}, which is a recursive algorithm that can compile any quantum gate with a circuit depth of $\log^c(1/\epsilon_{F})$, where $c\approx 4$ is a constant. Later, this constant was improved to $c \approx 1.44$~\cite{Kuperberg2023}. Solovay-Kitaev algorithm requires the gate set to be symmetric, i.e., the inverse of each gate is also easily accessible, which the proposed universal gate set $\{H, T, \CNOT\}$ does satisfy.
However, in the scattering based quantum computing framework, the inverse of a gate is not available.
In the latest progress of inverse-free Solovay-Kitaev algorithm~\cite{Bouland2021},
researchers presented a polylog algorithm which also works for gate sets without inverse.
However, the polynomial order $c \approx 12$ for two-qubit gate compiling is too large to be practically useful.
As shown in \Cref{fig:scaling}, when we set the target gate error to $10^{-3}$ as indicated by the dashed horizontal line,
most polylog algorithms do not even outperform the linear algorithm.
A new compiling method that can achieve a certain precision with a reasonable circuit depth is needed.

\begin{figure}[t]
    \centering
    \includegraphics[width=0.9\columnwidth, trim={0.3cm 0cm 0cm 0cm}, clip]{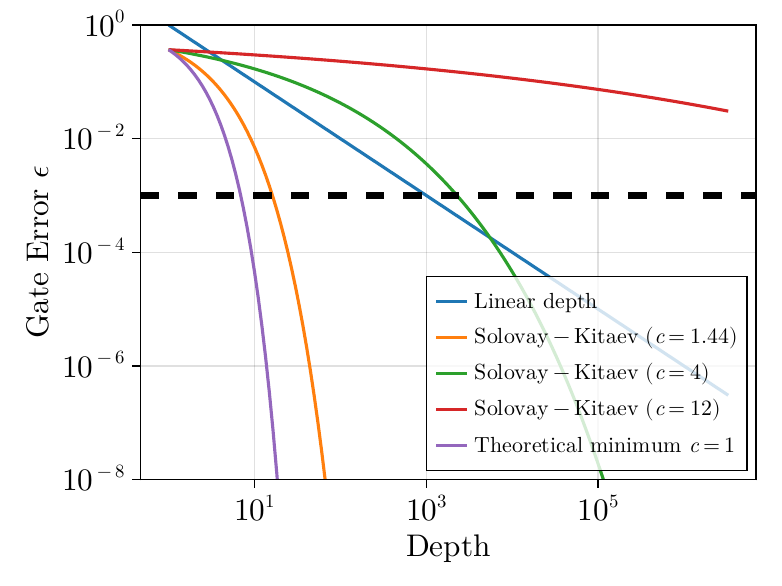}
    \caption{The circuit depth $d$ to achieve gate precision $\epsilon$ in different compiling methods, which includes the linear depth compiling from Ref.~\cite{Deutsch1995}, the Solovay-Kitaev's polylog depth compiling with polynomial order $c=4$ (inverse required)~\cite{Dawson2005} and its variant with $c=1.44$~\cite{Kuperberg2023}, and the inverse free Solovay-Kitaev algorithm with $c=12$~\cite{Bouland2021}.
    }\label{fig:scaling}
\end{figure}
\paragraph{Theoretical minimum circuit depth --}
Consider using the gate set in \Cref{eq:four-unitaries} to compile an arbitrary two-qubit gate. The compiling ability of a depth $d$ gate sequence can be measured by the \emph{mesh size}
\begin{equation}
    \epsilon_{F}^*(d) = \max_{\Gamma \in \SU(N)}\min_{D \in \mathcal{U}^d}\|D-\Gamma\|,
\end{equation}
which is determined by the hardest target gate to compile. The smaller the $\epsilon_F^*$ is, the more compiling ability the gate sequence has.
To achieve a target mesh size $\epsilon_F^*$, the theoretical minimum circuit depth scales as $d\approx A\log(1/\epsilon_{F}^*)$~\cite{harrow2002efficient} for some constant factor $A$, which can be achieved when the generated unitaries are uniformly distributed on the $\SU(N)$ manifold.
The linear log scaling is assured by the Lie group random walk theory~\cite{varju2012random}. As we present in \Cref{sec:appendix-gap}, the constant factor $A$ depends on the dimension $N$ of the unitary group and the spectral gap of the gate set $\mathcal{U}$. The larger the gap, the smaller $A$ is. And, it is believed the gap is non-zero for a universal gate set $\mathcal{U}$, providing a finite upper bound to $A$. 

The minimum possible $A$, $A^*$, is only relevant to the dimension $N$ and $|\mathcal{U}|$. To show this, we define $\epsilon^*_{F}$-ball of a unitary as the subset of $\SU(N)$ that containing all unitary matrices with distance less than $\epsilon^*_{F}$ to it.
Let the gate set generated by a depth $d$ gate sequence be $\mathcal{U}^d$, we have $|\mathcal{U}^d| = 4^d$. An arbitrary gate can be compiled to precision $\epsilon_F^*$ if the $\epsilon_F^*$-balls of all the generated unitaries can cover the whole $\SU(N)$ manifold.
Considering $\SU(N)$ as a $(N^2-1)$-dimensional manifold, the volume of the $\epsilon^*_{F}$-ball is equal to $ C {\epsilon^*_{F}}^{N^2-1}$ for some constant $C$.
To cover the whole $\SU(N)$, we have $|\mathcal{U}|^d C {\epsilon^*_{F}}^{N^2-1}\geq 1$, then we have
\begin{equation}
    d\geq (N^2-1)\log_{|\mathcal{U}|}\frac{1}{\epsilon^*_{F}}+\log_{|\mathcal{U}|} C,
\end{equation}
which is consistent with Ref.~\cite{harrow2002efficient}.

\begin{figure}[t]
    \centering
    \includegraphics[width=\columnwidth, trim={2cm 7cm 4cm 1.5cm}, clip]{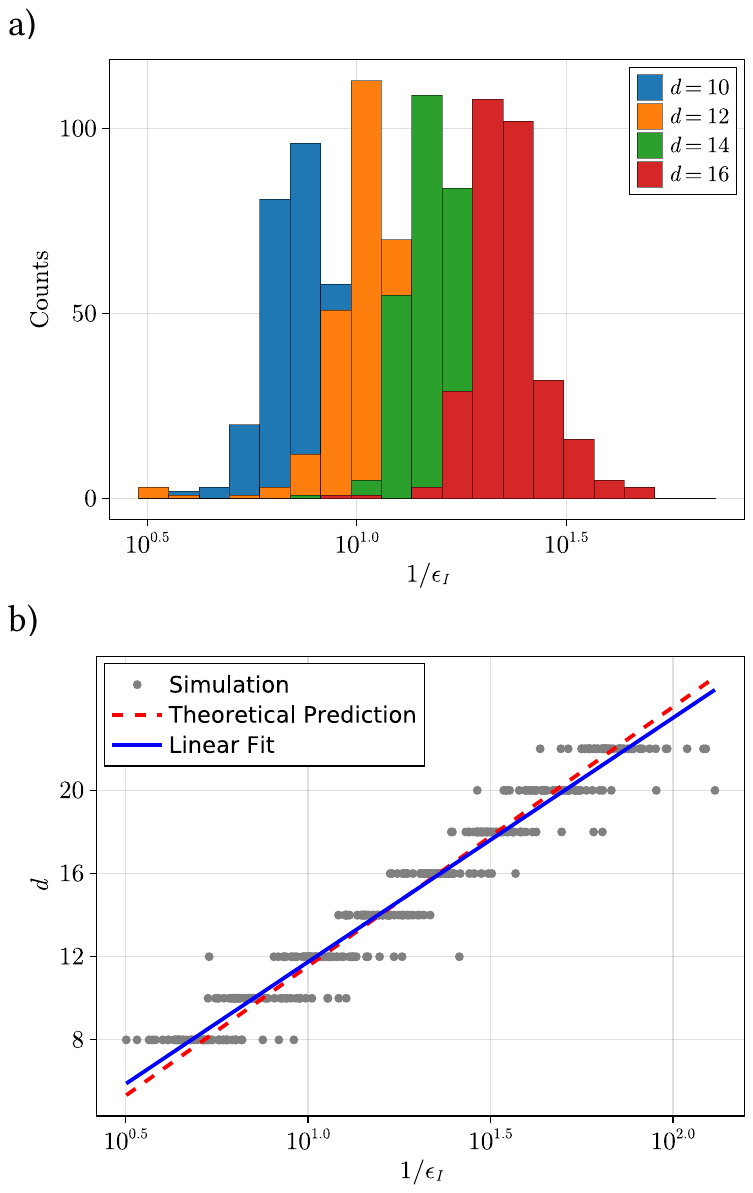}
    \caption{Compiling $\CNOT$ gate with one arbitrary gate and its variants in \Cref{eq:four-unitaries}. (a) Infidelity distribution of compilation with $300$ different arbitrary gates. The $x$-axis is the inverse of infidelity, $\epsilon_I$, and the $y$-axis is the number of gate sets with different infidelity. (b) The inverse of infidelity $1/\epsilon_I$ v.s. the circuit depth $d$. The gray dots are generated from $50$ numerical experiments. The blue line stands for the linear fitting of the experimental data. The red dashed line stands for the theoretical estimation in \Cref{eq:depth-epsilon}.} \label{fig:numericalresult}
\end{figure}

The Frobenius norm induced metric is not a perfect metric for quantum compilation since it does not consider the global phase factor.
In quantum mechanics, two unitaries that only differ by a global phase factor are considered to be the same.
In our numerical experiments, we use the operator infidelity $\epsilon_{I}(U_1,U_2) = 1 - \frac{\sqrt{Tr(U_1^{\dagger}U_2)^2}}{N}$ as a metric of compiling performance.
A similar bound for the circuit depth can be derived for the infidelity:
\begin{equation}\label{eq:depth-epsilon}
    \begin{aligned}
        d\sim &\frac{1}{2}(N^2-1)\log_{|\mathcal{U}|}\frac{1}{\epsilon^*_{I}}+\log_{|\mathcal{U}|} C'\\
        =& A_{I}^*\log_{10}\frac{1}{\epsilon^*_{I}}+B.
    \end{aligned}
\end{equation}
For our case $N=4,|\mathcal{U}|=4$, then we have the minimum prefactor $A^*_I=\frac{N^2-1}{2\log_{10}|\mathcal{U}|}\approx 12.46$ for the infidelity. This follows from the relation $\epsilon_{I}(U_1,U_2) \leq \frac{1}{2N} \epsilon_{F}(U_1,U_2)^2$, and its proof can be found in \Cref{sec:appendix-infvsfro}.

\subsection{Numerical Methods and Results}
\label{sec:num-meth-results}

In the following, we show the optimal prefactor $A^*_I$ in~\Cref{eq:depth-epsilon} is achievable for most $\mathcal{U}$ composed of a single arbitrary gate and its variants. The numerical method, detailed in~\Cref{sec:appendix-mim}, is based on the meet-in-the-middle compiling method which can achieve the theoretical optimal circuit depth.  

The results for $\CNOT$ gate compilation are shown in \Cref{fig:numericalresult}. We confirmed the results for other gates are similar, hence not shown here. For each circuit depth $d$, we compile the $\CNOT$ gate with different basis gate sets, $\{U,U^{T},U^{S},U^{ST}\}$, and record the compiling infidelity $\epsilon_I$. Figure (a) shows the distribution of infidelity for $300$ different basis gate sets. Both the worst and median infidelities decreases exponentially with the circuit depth.
Figrue (b) shows the linear fit to experimental data (blue line). The slope of which matches well with the theoretically prediction (red dashed line) with prefactor $A_I^* \approx 12.46$. Similar results are observed for systems with neither transpose operation nor the inverse operation, which is detailed in \Cref{sec:appendix-twou}. 

The maximum depth of the reported data is $d=22$ with our improved meet-in-the-middle algorithm. Each compilation takes $1.369$ minutes on a node with all $32$ cores of an Intel \textregistered~Xeon \textregistered~Platinum 8358P at 2.6 GHz. $31.157$ days was estimated to reach $10^{-3}$ infidelity with $36$ depth of compilation on a cluster with $8192$ nodes. The estimated resource is within the reach of current technology~\cite{haner20175}. 

\section{Measurement-free quantum error correction}\label{sec:qec}
The gate compilation precision $10^{-3}$ achievable by the improved meet-in-the-middle method is still far from implementing useful quantum algorithms.
In the following, we show the further improvement on how higher gate precision can be achieved by quantum error correction, by treating the imperfection of gates as coherent errors~\cite{iverson2020coherence}.
A MF-QEC ~\cite{heussen2024measurement,veroni2024optimized} is employed, since the mid-circuit measurement may not be feasible in systems with limited degrees of freedom such as the scattering based quantum computing that we are discussing.

As shown in \Cref{fig:mfmain}, quantum error correction mainly consists of three stages: \emph{encoding}, \emph{syndrome extraction}, and \emph{error correction}.
The main difference between the MF-QEC and the traditional QEC is on the \emph{error correction} stage. 
Those ancilla qubits with error syndrome information are processed by quantum operations, and the feedback is done through multi-control gates.
Although the MF-QEC has some disadvantages, such as requiring non-Clifford quantum operations (the multi-control gates), more complex quantum circuit as well as more errors chances, our numerical experiments shows that the MF-QEC can still be beneficial for the scattering based quantum computing framework.
We use the measurement-free Shor's code~\cite{veroni2024optimized} as an example to demonstrate the error correction process. The details of the error correction circuit is in \Cref{sec:appendix-mfqec}.

\begin{figure}
 
\tikzset{every picture/.style={line width=0.75pt}} 

\begin{tikzpicture}[x=0.75pt,y=0.75pt,yscale=-0.8,xscale=0.8]

\draw    (29,40) -- (57.5,40) ;
\draw   (110,32.83) .. controls (110,28.23) and (113.73,24.5) .. (118.33,24.5) -- (178.06,24.5) .. controls (182.66,24.5) and (186.39,28.23) .. (186.39,32.83) -- (186.39,86.17) .. controls (186.39,90.77) and (182.66,94.5) .. (178.06,94.5) -- (118.33,94.5) .. controls (113.73,94.5) and (110,90.77) .. (110,86.17) -- cycle ;
\draw    (30,82) -- (109.5,82) ;
\draw    (187,40) -- (322.5,40) ;
\draw    (187,82) -- (200,82) ;
\draw   (199.97,65.22) .. controls (199.97,62.92) and (201.84,61.06) .. (204.14,61.06) -- (230.8,61.06) .. controls (233.11,61.06) and (234.97,62.92) .. (234.97,65.22) -- (234.97,91.89) .. controls (234.97,94.19) and (233.11,96.06) .. (230.8,96.06) -- (204.14,96.06) .. controls (201.84,96.06) and (199.97,94.19) .. (199.97,91.89) -- cycle ;
\draw    (118.17,39.5) -- (118.17,69.83)(115.17,39.5) -- (115.17,69.83) ;
\draw [shift={(116.67,78.83)}, rotate = 270] [fill={rgb, 255:red, 0; green, 0; blue, 0 }  ][line width=0.08]  [draw opacity=0] (10.72,-5.15) -- (0,0) -- (10.72,5.15) -- (7.12,0) -- cycle    ;
\draw    (216.47,84.56) -- (224.86,71.44) ;
\draw [shift={(226.47,68.91)}, rotate = 122.58] [fill={rgb, 255:red, 0; green, 0; blue, 0 }  ][line width=0.08]  [draw opacity=0] (5.36,-2.57) -- (0,0) -- (5.36,2.57) -- cycle    ;
\draw  [draw opacity=0] (204.79,84) .. controls (206.13,80.31) and (211.6,77.54) .. (218.14,77.54) .. controls (225.33,77.54) and (231.22,80.88) .. (231.76,85.12) -- (218.14,85.76) -- cycle ; \draw   (204.79,84) .. controls (206.13,80.31) and (211.6,77.54) .. (218.14,77.54) .. controls (225.33,77.54) and (231.22,80.88) .. (231.76,85.12) ;  
\draw    (237.86,82) -- (255.14,82) ;
\draw [shift={(258.14,82)}, rotate = 180.6] [fill={rgb, 255:red, 0; green, 0; blue, 0 }  ][line width=0.08]  [draw opacity=0] (8.93,-4.29) -- (0,0) -- (8.93,4.29) -- cycle    ;
\draw   (58,28.29) .. controls (58,25.99) and (59.87,24.12) .. (62.17,24.12) -- (90.46,24.12) .. controls (92.76,24.12) and (94.63,25.99) .. (94.63,28.29) -- (94.63,54.95) .. controls (94.63,57.25) and (92.76,59.12) .. (90.46,59.12) -- (62.17,59.12) .. controls (59.87,59.12) and (58,57.25) .. (58,54.95) -- cycle ;
\draw    (94.5,40) -- (109.5,40) ;
\draw   (323,26.17) .. controls (323,23.87) and (324.87,22) .. (327.17,22) -- (359.23,22) .. controls (361.53,22) and (363.39,23.87) .. (363.39,26.17) -- (363.39,52.83) .. controls (363.39,55.13) and (361.53,57) .. (359.23,57) -- (327.17,57) .. controls (324.87,57) and (323,55.13) .. (323,52.83) -- cycle ;
\draw    (343.5,82) -- (343.5,64.58) ;
\draw [shift={(343.5,61.58)}, rotate = 90] [fill={rgb, 255:red, 0; green, 0; blue, 0 }  ][line width=0.08]  [draw opacity=0] (8.93,-4.29) -- (0,0) -- (8.93,4.29) -- cycle    ;
\draw    (325.6,82) -- (343.5,82) ;
\draw    (35.71,87.75) -- (44.57,77.86) ;
\draw    (35.71,46.75) -- (44.57,36.86) ;
\draw    (188.67,211) -- (210.33,211) ;
\draw   (210.67,200.28) .. controls (210.67,197.98) and (212.53,196.11) .. (214.83,196.11) -- (248.7,196.11) .. controls (251,196.11) and (252.86,197.98) .. (252.86,200.28) -- (252.86,226.94) .. controls (252.86,229.24) and (251,231.11) .. (248.7,231.11) -- (214.83,231.11) .. controls (212.53,231.11) and (210.67,229.24) .. (210.67,226.94) -- cycle ;
\draw    (253.67,211) -- (269.67,211) ;
\draw  [fill={rgb, 255:red, 0; green, 0; blue, 0 }  ,fill opacity=1 ] (227.4,253.78) .. controls (227.4,251.29) and (229.41,249.28) .. (231.9,249.28) .. controls (234.39,249.28) and (236.4,251.29) .. (236.4,253.78) .. controls (236.4,256.26) and (234.39,258.28) .. (231.9,258.28) .. controls (229.41,258.28) and (227.4,256.26) .. (227.4,253.78) -- cycle ;
\draw    (231.9,253.78) -- (231.9,231) ;
\draw   [shift={(-3.5,0)}](365.8,190.68) .. controls (365.77,186.01) and (363.43,183.69) .. (358.76,183.71) -- (299.76,184.02) .. controls (293.09,184.05) and (289.75,181.74) .. (289.73,177.07) .. controls (289.75,181.74) and (286.43,184.09) .. (279.76,184.13)(282.76,184.11) -- (220.76,184.43) .. controls (216.09,184.46) and (213.77,186.8) .. (213.8,191.47) ;
\draw   (318.07,200.28) .. controls (318.07,197.98) and (319.93,196.11) .. (322.23,196.11) -- (359.36,196.11) .. controls (361.66,196.11) and (363.53,197.98) .. (363.53,200.28) -- (363.53,226.94) .. controls (363.53,229.24) and (361.66,231.11) .. (359.36,231.11) -- (322.23,231.11) .. controls (319.93,231.11) and (318.07,229.24) .. (318.07,226.94) -- cycle ;
\draw  [fill={rgb, 255:red, 0; green, 0; blue, 0 }  ,fill opacity=1 ] (335.7,253.78) .. controls (335.7,251.29) and (337.71,249.28) .. (340.2,249.28) .. controls (342.69,249.28) and (344.7,251.29) .. (344.7,253.78) .. controls (344.7,256.26) and (342.69,258.28) .. (340.2,258.28) .. controls (337.71,258.28) and (335.7,256.26) .. (335.7,253.78) -- cycle ;
\draw    (340.2,253.78) -- (340.2,231) ;
\draw    (302.47,213.81) -- (317.47,213.9) ;
\draw  [fill={rgb, 255:red, 0; green, 0; blue, 0 }  ,fill opacity=1 ] (278.4,253.5) .. controls (278.4,252.67) and (279.07,252) .. (279.9,252) .. controls (280.73,252) and (281.4,252.67) .. (281.4,253.5) .. controls (281.4,254.33) and (280.73,255) .. (279.9,255) .. controls (279.07,255) and (278.4,254.33) .. (278.4,253.5) -- cycle ;
\draw  [fill={rgb, 255:red, 0; green, 0; blue, 0 }  ,fill opacity=1 ] (291.73,253.5) .. controls (291.73,252.67) and (292.4,252) .. (293.23,252) .. controls (294.06,252) and (294.73,252.67) .. (294.73,253.5) .. controls (294.73,254.33) and (294.06,255) .. (293.23,255) .. controls (292.4,255) and (291.73,254.33) .. (291.73,253.5) -- cycle ;
\draw  [fill={rgb, 255:red, 0; green, 0; blue, 0 }  ,fill opacity=1 ] (285,253.5) .. controls (285,252.67) and (285.67,252) .. (286.5,252) .. controls (287.33,252) and (288,252.67) .. (288,253.5) .. controls (288,254.33) and (287.33,255) .. (286.5,255) .. controls (285.67,255) and (285,254.33) .. (285,253.5) -- cycle ;

\draw  [fill={rgb, 255:red, 0; green, 0; blue, 0 }  ,fill opacity=1 ] (278.4,213.33) .. controls (278.4,212.5) and (279.07,211.83) .. (279.9,211.83) .. controls (280.73,211.83) and (281.4,212.5) .. (281.4,213.33) .. controls (281.4,214.16) and (280.73,214.83) .. (279.9,214.83) .. controls (279.07,214.83) and (278.4,214.16) .. (278.4,213.33) -- cycle ;
\draw  [fill={rgb, 255:red, 0; green, 0; blue, 0 }  ,fill opacity=1 ] (291.73,213.33) .. controls (291.73,212.5) and (292.4,211.83) .. (293.23,211.83) .. controls (294.06,211.83) and (294.73,212.5) .. (294.73,213.33) .. controls (294.73,214.16) and (294.06,214.83) .. (293.23,214.83) .. controls (292.4,214.83) and (291.73,214.16) .. (291.73,213.33) -- cycle ;
\draw  [fill={rgb, 255:red, 0; green, 0; blue, 0 }  ,fill opacity=1 ] (285,213.33) .. controls (285,212.5) and (285.67,211.83) .. (286.5,211.83) .. controls (287.33,211.83) and (288,212.5) .. (288,213.33) .. controls (288,214.16) and (287.33,214.83) .. (286.5,214.83) .. controls (285.67,214.83) and (285,214.16) .. (285,213.33) -- cycle ;

\draw    (268.7,254) -- (189,254) ;
\draw    (363.83,254) -- (301.33,254) ;

\draw  [fill={rgb, 255:red, 255; green, 255; blue, 255 }  ,fill opacity=1 ] (282.49,67.48) .. controls (282.49,64.48) and (284.92,62.04) .. (287.92,62.04) -- (312.98,62.04) .. controls (315.99,62.04) and (318.42,64.48) .. (318.42,67.48) -- (318.42,83.79) .. controls (318.42,86.8) and (315.99,89.23) .. (312.98,89.23) -- (287.92,89.23) .. controls (284.92,89.23) and (282.49,86.8) .. (282.49,83.79) -- cycle ;
\draw  [fill={rgb, 255:red, 0; green, 0; blue, 0 }  ,fill opacity=1 ] (286.46,66.18) -- (314.45,66.18) -- (314.45,85.1) -- (286.46,85.1) -- cycle ;
\draw  [fill={rgb, 255:red, 255; green, 255; blue, 255 }  ,fill opacity=1 ] (292.03,94.32) -- (293.6,89.07) -- (307.94,89.07) -- (309.52,94.32) -- cycle ;
\draw  [color={rgb, 255:red, 0; green, 0; blue, 0 }  ,draw opacity=1 ][fill={rgb, 255:red, 255; green, 255; blue, 255 }  ,fill opacity=1 ] (288.31,92.92) .. controls (288.31,92.55) and (288.6,92.25) .. (288.97,92.25) -- (312.98,92.25) .. controls (313.35,92.25) and (313.65,92.55) .. (313.65,92.92) -- (313.65,94.92) .. controls (313.65,95.29) and (313.35,95.59) .. (312.98,95.59) -- (288.97,95.59) .. controls (288.6,95.59) and (288.31,95.29) .. (288.31,94.92) -- cycle ;
\draw   (266.27,64.97) .. controls (266.27,63.35) and (267.58,62.04) .. (269.19,62.04) -- (277.97,62.04) .. controls (279.59,62.04) and (280.9,63.35) .. (280.9,64.97) -- (280.9,92.7) .. controls (280.9,94.31) and (279.59,95.62) .. (277.97,95.62) -- (269.19,95.62) .. controls (267.58,95.62) and (266.27,94.31) .. (266.27,92.7) -- cycle ;
\draw  [shift={(0.15,0)}] [fill={rgb, 255:red, 0; green, 0; blue, 0 }  ,fill opacity=1 ] (267.54,65.53) .. controls (267.54,65.35) and (267.68,65.21) .. (267.86,65.21) -- (278.99,65.21) .. controls (279.16,65.21) and (279.31,65.35) .. (279.31,65.53) -- (279.31,66.48) .. controls (279.31,66.66) and (279.16,66.8) .. (278.99,66.8) -- (267.86,66.8) .. controls (267.68,66.8) and (267.54,66.66) .. (267.54,66.48) -- cycle ;
\draw  [shift={(0.15,0)}] [fill={rgb, 255:red, 0; green, 0; blue, 0 }  ,fill opacity=1 ] (267.54,69.52) .. controls (267.54,69.34) and (267.68,69.2) .. (267.86,69.2) -- (278.99,69.2) .. controls (279.16,69.2) and (279.31,69.34) .. (279.31,69.52) -- (279.31,70.47) .. controls (279.31,70.64) and (279.16,70.79) .. (278.99,70.79) -- (267.86,70.79) .. controls (267.68,70.79) and (267.54,70.64) .. (267.54,70.47) -- cycle ;
\draw  [fill={rgb, 255:red, 0; green, 0; blue, 0 }  ,fill opacity=1 ] (269.13,81.22) .. controls (269.13,81.08) and (269.24,80.96) .. (269.38,80.96) -- (277.78,80.96) .. controls (277.92,80.96) and (278.03,81.08) .. (278.03,81.22) -- (278.03,81.98) .. controls (278.03,82.12) and (277.92,82.24) .. (277.78,82.24) -- (269.38,82.24) .. controls (269.24,82.24) and (269.13,82.12) .. (269.13,81.98) -- cycle ;
\draw  [fill={rgb, 255:red, 0; green, 0; blue, 0 }  ,fill opacity=1 ] (272.15,86.29) .. controls (272.15,85.54) and (272.76,84.94) .. (273.5,84.94) .. controls (274.25,84.94) and (274.85,85.54) .. (274.85,86.29) .. controls (274.85,87.04) and (274.25,87.64) .. (273.5,87.64) .. controls (272.76,87.64) and (272.15,87.04) .. (272.15,86.29) -- cycle ;

\draw    (32.17,211) -- (59.67,211) ;
\draw   (112.17,202.39) .. controls (112.17,197.79) and (115.9,194.06) .. (120.5,194.06) -- (180.23,194.06) .. controls (184.83,194.06) and (188.56,197.79) .. (188.56,202.39) -- (188.56,255.73) .. controls (188.56,260.33) and (184.83,264.06) .. (180.23,264.06) -- (120.5,264.06) .. controls (115.9,264.06) and (112.17,260.33) .. (112.17,255.73) -- cycle ;
\draw    (32.17,254) -- (111.67,254) ;
\draw    (120.33,209.06) -- (120.33,239.39)(117.33,209.06) -- (117.33,239.39) ;
\draw [shift={(118.83,248.39)}, rotate = 270] [fill={rgb, 255:red, 0; green, 0; blue, 0 }  ][line width=0.08]  [draw opacity=0] (10.72,-5.15) -- (0,0) -- (10.72,5.15) -- (7.12,0) -- cycle    ;
\draw   (60.17,197.85) .. controls (60.17,195.55) and (62.03,193.68) .. (64.33,193.68) -- (92.63,193.68) .. controls (94.93,193.68) and (96.79,195.55) .. (96.79,197.85) -- (96.79,224.51) .. controls (96.79,226.82) and (94.93,228.68) .. (92.63,228.68) -- (64.33,228.68) .. controls (62.03,228.68) and (60.17,226.82) .. (60.17,224.51) -- cycle ;
\draw    (96.67,211) -- (111.67,211) ;
\draw    (37.88,257.31) -- (46.74,247.42) ;
\draw    (37.88,216.31) -- (46.74,206.42) ;

\draw (3,-5) node [anchor=north west][inner sep=0.75pt]  [font=\large] [align=left] {(a)};
\draw (2,32) node [anchor=north west][inner sep=0.75pt]    {$|\psi \rangle $};
\draw (2,71) node [anchor=north west][inner sep=0.75pt]    {$|0\rangle ^{\otimes l}$};
\draw (121,41.83) node [anchor=north west][inner sep=0.75pt]   [align=left] {\scriptsize Syndrome\\ \scriptsize Extraction};
\draw (58.5,35) node [anchor=north west][inner sep=0.75pt]   [align=left] {\scriptsize Error};
\draw (336,33) node [anchor=north west][inner sep=0.75pt]   [align=left] {\scriptsize $F$};
\draw (44.14,65.86) node [anchor=north west][inner sep=0.75pt]    {$l$};
\draw (43.34,25.06) node [anchor=north west][inner sep=0.75pt]    {$n$};
\draw (223,206) node [anchor=north west][inner sep=0.75pt]   [align=left] {\scriptsize $F _1$};
\draw (3,160) node [anchor=north west][inner sep=0.75pt]  [font=\large] [align=left] {(b)};
\draw (230,101) node [anchor=north west][inner sep=0.75pt]   [align=left] {Classical Computer};
\draw (330,206) node [anchor=north west][inner sep=0.75pt]   [align=left] {\scriptsize $F _{m}$};
\draw (210,154.33) node [anchor=north west][inner sep=0.75pt]   [align=left] {Controlled Corrections};
\draw (2,201.56) node [anchor=north west][inner sep=0.75pt]    {$|\psi \rangle $};
\draw (2,240) node [anchor=north west][inner sep=0.75pt]    {$|0\rangle ^{\otimes l}$};
\draw (123.17,211.39) node [anchor=north west][inner sep=0.75pt]   [align=left] {\scriptsize Syndrome\\ \scriptsize Extraction};
\draw (60.5,204.23) node [anchor=north west][inner sep=0.75pt]   [align=left] {\scriptsize Error};
\draw (46.31,235.42) node [anchor=north west][inner sep=0.75pt]    {$l$};
\draw (45.51,194.62) node [anchor=north west][inner sep=0.75pt]    {$n$};

\end{tikzpicture}
\caption{The (a) traditional and (b) measurement-free quantum error correction protocol. Instead of analyzing the error syndromes with a classical device and sending the corresponding correction operation $F$ back, the measurement-free QEC detects all possible error types with a quantum circuit and sends the error correction signals through controlled gates. $F_{k}$ is applied only if the $k$-th signal is true. \Cref{sec:appendix-mfqec} shows an example of measurement-free QEC with 3-qubit repetition code.}\label{fig:mfmain}
\end{figure}
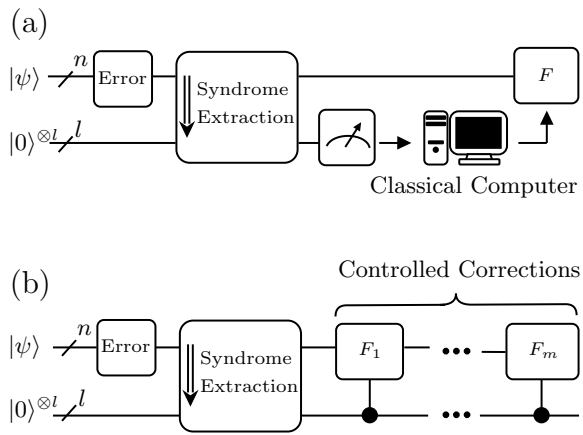

The Shor code is a nine-qubit CSS code~\cite{shor1995scheme}, working by concatenating each qubit of a phase-flip with a bit-flip repetition code. The bit-flip stabilizers are $\{Z_iZ_j \mid (i,j) \in \{(1, 2), (2, 3), (4, 5), (5, 6), (7, 8), (8, 9)\}\}$ and the phase-flip stabilizers for the second layer encoding are $\overline{X_1}\overline{X_2}$ and $\overline{X_2}\overline{X_3}$. Here $\overline{X}$ stands for the logical $X$ operator acting on the logical qubits. The two Shor code states are
$$
\begin{aligned}
    |\overline{0}\rangle &= |\overline{+}\rangle _b^{\otimes 3}&= \frac{1}{2\sqrt{2}}(|000\rangle + |111\rangle)^{\otimes 3},\\
    |\overline{1}\rangle &= |\overline{-}\rangle _b^{\otimes 3}&= \frac{1}{2\sqrt{2}}(|000\rangle - |111\rangle)^{\otimes 3}.
\end{aligned}
$$

Comparing with the traditional QEC, the MF-QEC protocol requires $4$ extra qubits and $12$ multi-control gates.
The $4$ extra qubits are used to store the error syndromes of $4$ extra stabilizers that are linearly dependent with the original stabilizers, which are later used for controlled correction.
The $12$ multi-control gates are for correcting errors, which may cause more error to the logical qubits.
The threshold of the MF-QEC protocal may increase or decrease depending on the error model of the physical platform.
If the error of the gate is much smaller than the error of the measurement, the threshold of the MF-QEC may be higher than that of the traditional QEC.
To analyze the error correction performance, we repeatedly simulate the Shor code circuit with different gate compilation errors.
Only single qubit errors are considered since the Shor code has a code distance of $3$, which can only correct single qubit errors.
The test circuit consists of two (logical) $X$ gates followed by the error correction circuit. By including all ancilla qubits, the total number of qubits is $21$.
Since coherent error gates are no longer Clifford gates, efficient classical simulation is not available. We have to resort to the full-amplitude simulation, which is perform with the Julia~\cite{Bezanson2017} package \texttt{Yao.jl}~\cite{Luo2020}.
The results are shown in \Cref{fig:qec_res}. We randomly generate $10^4$ samples with gate infidelities from $10^{-6}$ to $0.5$, and compare the error probabilities of different error correction protocals.
Three different protocals are considered: directly operating on physical qubits, operating logical qubits without error correction, and operating logical qubits followed by the MF-QEC.
The error probability distribution indicates an upper linear bound on the log-log scale figure.
The MF-QEC (green line) significantly reduces the upper bound of error probability.
With a gate infidelity $10^{-3}$, the error probability after a single round of MF-QEC is only $2\times 10^{-5}$ (black horizontal line), which is much lower than those without QEC.
Hence, we conclude that the MF-QEC protocol can effectively reduce the compilation errors on single qubits and improve the gate precision.
\begin{figure}
    \centering
    \includegraphics[width=0.95\columnwidth]{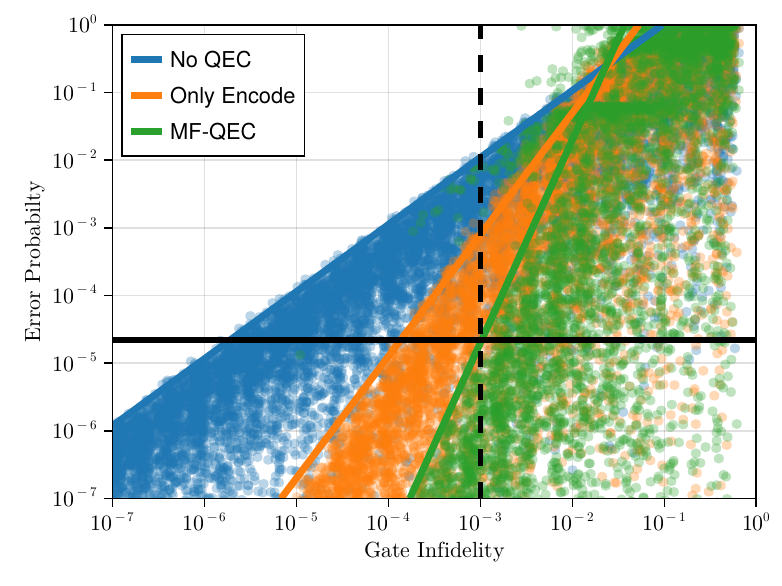}
    \caption{The error probability versus gate infidelity.
    Three different colors represent different QEC protocals: directly operating physical qubits (No QEC), applying logical $X$ gates on Shor code without correction (Only Encode), and MF-QEC on Shor code (MF-QEC). Each category contains $10^4$ instances (dots).
    As indicated by the black lines, the error probability for a $10^{-3}$ gate infidelity is reduced to $2\times 10^{-5}$ by the MF-QEC.}
    \label{fig:qec_res}
\end{figure}

\section{Discussion}
In this letter, we proposed a scheme for utilizing a single arbitrary gate to achieve fault-tolerant quantum computation.
An improved brute-force search algorithm is proposed to find the optimal compiling sequence for a given gate.
With this scheme, the gate precision of $10^{-3}$ is achievable with a circuit depth of $\sim 36$.
The remaining compilation error is then treated as coherent error, and can be efficiently corrected by the MF-QEC protocol.
These results open a door for utilizing natural compounds for quantum computing.
While showing great potential, our work also raises multiple challenges for future research of single arbitrary gate universal quantum computing.
The first one is how to concatenate the error correction code for finally achieving fault tolerance. The current scheme contains some non-transversal gates, which makes the concatenation challenging. 
The second one is how to handle multi-qubit errors in the MF-QEC protocol.

\begin{acknowledgments}
    The authors thank Xun Gao for sharing his valuable insights with us, Jian-Xin Chen, Xiong-Feng Ma, Bei Zeng, Yi-Nan Li and Dong Liu for helpful discussion on quantum circuit compilation and quantum error correction and Xia-Kun Chu, Long-Li Zheng and Yun-Fan Qiu for helpful discussion on electron transportation properties in polymers.
    This work is partially funded by the National Key R\&D Program of China (Grant No. 2024YFE0102500), National Natural Science Foundation of China (No. 12404568) and the Guangzhou Municipal Science and Technology Project (No. 2023A03J00904).
\end{acknowledgments}

\bibliography{refs}

\onecolumngrid
\newpage
\appendix

\section{Linear Relationship between Circuit Depth and Log-inverse Gate Imprecision}
\label{sec:appendix-gap}
This appendix mainly introduces the linear relationship between the circuit depth and the log-inverse gate imprecision $\log_{10}(1/\epsilon)$ theoretically. To avoid the conflict of notation, we will use $l$ to denote the circuit depth in this appendix. The main references of the gap theories are ~\cite{slowik2023calculable,varju2012random}.

\subsection{Spectral gap}
Let $G$ be a semi-simple compact Lie group with a uniform probability measure $\mu$ and a bi-invariant metric $d$. In our situation, $G$ is the unitary group $\SU(4)$,
$\mu$ is the Haar measure and $d$ is the metric induced by the Frobenius norm. Let $\mathcal{U}$ be a finite generating set of $G$. We suppose that the identity element $I$ is in $\mathcal{U}$.
All possible unitaries generated by a length $l$ quantum circuit are denoted as $\mathcal{U}^l$, which is the $l$-length product set of $\mathcal{U}$.
$$
\mathcal{U}^l=\{g_1g_2...g_l:g_i \in \mathcal{U},i=1,2,\dots,l\}.
$$
To estimate the coverage of $\mathcal{U}^l$ on $G$, we place a ``ball'' centered at each element in $\mathcal{U}^l$ with radius $r$ and check weather the union of these balls covers $G$.
Here, the ``ball'' is defined as the neighborhood of an element $g^* \in G$ is $B_r(g^*)=\{g\in G: d(g,g^*)<r \}$ with a radius $r$. If the union of these balls with radius $\epsilon$ covers $G$, we say that $\mathcal{U}^l$ is an $\epsilon$-net of $G$. 

In the main text, we calculate the volume of the balls to get a lower bound of the number of the balls, which is also a lower bound of the circuit depth $l$. The volume of $B_r(g^*)$ is $V(B_r(g^*))=C_V\cdot r^{\dim(G)}$, where $C_V$ is a constant depending on the group $G$ and the dimension of $G$ is $\dim(G)=\dim(\SU(4))= 15$ .
Here, we want to get an upper bound by comparing the difference between the uniform measure and the measure endowed by $\mathcal{U}^l$.
We use $\nu_{\mathcal{U}^l}$ to denote the discrete measure endowed by $\mathcal{U}^l$, i.e. 
$$
\nu_{\mathcal{U}^l}(g) =\left\{\begin{aligned}
    \frac{1}{\lvert \mathcal{U}^l\rvert} &\hspace*{1cm} g\in \mathcal{U}^l;\\
    0 & \hspace*{1cm} g \notin \mathcal{U}^l.
    \end{aligned}\right. 
$$
To compare the difference between a continuous measure $\mu$ and a discrete measure $\nu$, we need a metric about the closeness of two measures. We find out that the spectral gap is a good metric and it can be used to bound $\epsilon$ further. The spectral gap is defined on an operator norm of operators that act on the function space $L^2(G)$, which is the space of all the square integrable complex functions on $G$ with respect to $\mu$. 
$$
L^2(G) = \{f:G \rightarrow \mathbb{C} \text{ s.t. } \int_G \lvert f(g) \rvert^2 d\mu(g) < \infty \}.
$$
The inner product on $L^2(G)$ is defined as 
$$
\langle f_1,f_2 \rangle=\int_G f_1(g)^*f_2(g)d\mu(g),
$$
which is consistent with the 2-norm, i.e. $\langle f,f \rangle=\lVert f \rVert^2_2$. For any $g \in G$, we define the translation operator on $L^2(G)$ by 
$$
T_g(f)(h)=f(g^{-1}h).
$$
We define $\mathcal{T}_\nu$ as a linear operator on $L^2(G)$, which takes a function $f$ to its translational average with respect to $\nu$:
\begin{equation*}\label{eq:translation-average}
\begin{aligned}
    \mathcal{T}_{\nu } (f)(h) & = \sum _{g\in G} \nu (g)T_{g}(f)(h) =\sum _{g\in G} \nu (g)f(g^{-1} h) & \text{(discrete)}\\
     & =\int_G f(g^{-1} h)d\nu ( g) & \text{(continuous)}.
    \end{aligned}    
\end{equation*}

\Cref{fig:appendix-translation} shows the function defined on the group $G$ and how the translation and average operator work on it.
\begin{figure}
\tikzset{every picture/.style={line width=0.75pt}} 

\centerline{
\begin{tikzpicture}[x=0.75pt,y=0.75pt,yscale=-1,xscale=1]

\draw  [draw opacity=0][fill={rgb, 255:red, 80; green, 227; blue, 194 }  ,fill opacity=1 ][line width=1.5]  (540.41,63.74) .. controls (555.75,63.68) and (568.34,69.04) .. (573.88,79.54) .. controls (583.86,98.45) and (567.2,126.83) .. (536.67,142.94) -- (518.61,108.71) -- cycle ; \draw  [color={rgb, 255:red, 74; green, 144; blue, 226 }  ,draw opacity=1 ][line width=1.5]  (540.41,63.74) .. controls (555.75,63.68) and (568.34,69.04) .. (573.88,79.54) .. controls (583.86,98.45) and (567.2,126.83) .. (536.67,142.94) ;  
\draw  [draw opacity=0][fill={rgb, 255:red, 80; green, 227; blue, 194 }  ,fill opacity=1 ][line width=1.5]  (479,113.02) .. controls (479,113.02) and (479,113.02) .. (479,113.02) .. controls (479,78.5) and (496.33,50.52) .. (517.7,50.52) .. controls (526.71,50.52) and (534.99,55.49) .. (541.57,63.82) -- (517.7,113.02) -- cycle ; \draw  [color={rgb, 255:red, 74; green, 144; blue, 226 }  ,draw opacity=1 ][line width=1.5]  (479,113.02) .. controls (479,113.02) and (479,113.02) .. (479,113.02) .. controls (479,78.5) and (496.33,50.52) .. (517.7,50.52) .. controls (526.71,50.52) and (534.99,55.49) .. (541.57,63.82) ;  
\draw  [draw opacity=0][fill={rgb, 255:red, 80; green, 227; blue, 194 }  ,fill opacity=1 ][line width=1.5]  (207.1,30.7) .. controls (228.47,30.7) and (245.8,67.14) .. (245.8,112.08) .. controls (245.8,89.44) and (228.47,71.1) .. (207.1,71.1) .. controls (185.73,71.1) and (168.4,89.44) .. (168.4,112.08) .. controls (168.4,67.14) and (185.73,30.7) .. (207.1,30.7) -- cycle ;
\draw   (168.4,109.7) .. controls (168.4,88.33) and (185.73,71) .. (207.1,71) .. controls (228.47,71) and (245.8,88.33) .. (245.8,109.7) .. controls (245.8,131.07) and (228.47,148.4) .. (207.1,148.4) .. controls (185.73,148.4) and (168.4,131.07) .. (168.4,109.7) -- cycle ;
\draw  [fill={rgb, 255:red, 0; green, 0; blue, 0 }  ,fill opacity=1 ] (204.6,71.5) .. controls (204.6,70.12) and (205.72,69) .. (207.1,69) .. controls (208.48,69) and (209.6,70.12) .. (209.6,71.5) .. controls (209.6,72.88) and (208.48,74) .. (207.1,74) .. controls (205.72,74) and (204.6,72.88) .. (204.6,71.5) -- cycle ;

\draw  [draw opacity=0][line width=1.5]  (168,109.7) .. controls (168,109.7) and (168,109.7) .. (168,109.7) .. controls (168,65.83) and (185.33,30.27) .. (206.7,30.27) .. controls (228.07,30.27) and (245.4,65.83) .. (245.4,109.7) -- (206.7,109.7) -- cycle ; \draw  [color={rgb, 255:red, 74; green, 144; blue, 226 }  ,draw opacity=1 ][line width=1.5]  (168,109.7) .. controls (168,109.7) and (168,109.7) .. (168,109.7) .. controls (168,65.83) and (185.33,30.27) .. (206.7,30.27) .. controls (228.07,30.27) and (245.4,65.83) .. (245.4,109.7) ;  
\draw  [draw opacity=0][line width=1.5]  (245.14,105.71) .. controls (245.31,107.18) and (245.4,108.68) .. (245.4,110.2) .. controls (245.4,131.57) and (228.07,148.9) .. (206.7,148.9) .. controls (185.32,148.9) and (168,131.57) .. (168,110.2) .. controls (168,110.03) and (168,109.87) .. (168,109.7) -- (206.7,110.2) -- cycle ; \draw  [color={rgb, 255:red, 74; green, 144; blue, 226 }  ,draw opacity=1 ][line width=1.5]  (245.14,105.71) .. controls (245.31,107.18) and (245.4,108.68) .. (245.4,110.2) .. controls (245.4,131.57) and (228.07,148.9) .. (206.7,148.9) .. controls (185.32,148.9) and (168,131.57) .. (168,110.2) .. controls (168,110.03) and (168,109.87) .. (168,109.7) ;  

\draw  [fill={rgb, 255:red, 255; green, 255; blue, 255 }  ,fill opacity=1 ] (479,109.7) .. controls (479,88.33) and (496.33,71) .. (517.7,71) .. controls (539.07,71) and (556.4,88.33) .. (556.4,109.7) .. controls (556.4,131.07) and (539.07,148.4) .. (517.7,148.4) .. controls (496.33,148.4) and (479,131.07) .. (479,109.7) -- cycle ;
\draw  [fill={rgb, 255:red, 0; green, 0; blue, 0 }  ,fill opacity=1 ] (515.2,71.5) .. controls (515.2,70.12) and (516.32,69) .. (517.7,69) .. controls (519.08,69) and (520.2,70.12) .. (520.2,71.5) .. controls (520.2,72.88) and (519.08,74) .. (517.7,74) .. controls (516.32,74) and (515.2,72.88) .. (515.2,71.5) -- cycle ;
\draw  [fill={rgb, 255:red, 0; green, 0; blue, 0 }  ,fill opacity=1 ] (550.2,93.5) .. controls (550.2,92.12) and (551.32,91) .. (552.7,91) .. controls (554.08,91) and (555.2,92.12) .. (555.2,93.5) .. controls (555.2,94.88) and (554.08,96) .. (552.7,96) .. controls (551.32,96) and (550.2,94.88) .. (550.2,93.5) -- cycle ;
\draw   (30,109.7) .. controls (30,88.33) and (47.33,71) .. (68.7,71) .. controls (90.07,71) and (107.4,88.33) .. (107.4,109.7) .. controls (107.4,131.07) and (90.07,148.4) .. (68.7,148.4) .. controls (47.33,148.4) and (30,131.07) .. (30,109.7) -- cycle ;
\draw  [fill={rgb, 255:red, 0; green, 0; blue, 0 }  ,fill opacity=1 ] (66.2,71.5) .. controls (66.2,70.12) and (67.32,69) .. (68.7,69) .. controls (70.08,69) and (71.2,70.12) .. (71.2,71.5) .. controls (71.2,72.88) and (70.08,74) .. (68.7,74) .. controls (67.32,74) and (66.2,72.88) .. (66.2,71.5) -- cycle ;
\draw  [draw opacity=0][line width=1.5]  (536.67,142.94) .. controls (531.06,146.1) and (524.59,147.9) .. (517.7,147.9) .. controls (496.32,147.9) and (479,130.57) .. (479,109.2) .. controls (479,109.03) and (479,108.87) .. (479,108.7) -- (517.7,109.2) -- cycle ; \draw  [color={rgb, 255:red, 74; green, 144; blue, 226 }  ,draw opacity=1 ][line width=1.5]  (536.67,142.94) .. controls (531.06,146.1) and (524.59,147.9) .. (517.7,147.9) .. controls (496.32,147.9) and (479,130.57) .. (479,109.2) .. controls (479,109.03) and (479,108.87) .. (479,108.7) ;  
\draw  [draw opacity=0][fill={rgb, 255:red, 80; green, 227; blue, 194 }  ,fill opacity=1 ][line width=1.5]  (409.36,72.31) .. controls (419.38,91.19) and (395.32,123.57) .. (355.61,144.63) .. controls (375.61,134.03) and (383.7,110.12) .. (373.68,91.24) .. controls (363.66,72.36) and (339.33,65.65) .. (319.34,76.26) .. controls (359.04,55.2) and (399.34,53.43) .. (409.36,72.31) -- cycle ;
\draw   (321.35,75.19) .. controls (340.23,65.17) and (363.66,72.36) .. (373.68,91.24) .. controls (383.7,110.12) and (376.51,133.55) .. (357.63,143.56) .. controls (338.75,153.58) and (315.32,146.4) .. (305.31,127.52) .. controls (295.29,108.64) and (302.47,85.21) .. (321.35,75.19) -- cycle ;
\draw  [fill={rgb, 255:red, 0; green, 0; blue, 0 }  ,fill opacity=1 ] (372.07,89.27) .. controls (373.29,88.62) and (374.8,89.08) .. (375.45,90.3) .. controls (376.09,91.52) and (375.63,93.03) .. (374.41,93.68) .. controls (373.19,94.33) and (371.68,93.87) .. (371.03,92.65) .. controls (370.38,91.43) and (370.85,89.91) .. (372.07,89.27) -- cycle ;
\draw  [draw opacity=0][line width=1.5]  (321.17,74.84) .. controls (359.92,54.28) and (399.45,52.92) .. (409.47,71.8) .. controls (419.48,90.68) and (396.19,122.65) .. (357.44,143.21) -- (339.31,109.02) -- cycle ; \draw  [color={rgb, 255:red, 74; green, 144; blue, 226 }  ,draw opacity=1 ][line width=1.5]  (321.17,74.84) .. controls (359.92,54.28) and (399.45,52.92) .. (409.47,71.8) .. controls (419.48,90.68) and (396.19,122.65) .. (357.44,143.21) ;  
\draw  [draw opacity=0][line width=1.5]  (360.85,141.11) .. controls (359.62,141.95) and (358.34,142.73) .. (357,143.44) .. controls (338.12,153.46) and (314.69,146.28) .. (304.68,127.39) .. controls (294.66,108.51) and (301.84,85.09) .. (320.72,75.07) .. controls (320.87,74.99) and (321.02,74.91) .. (321.17,74.84) -- (338.86,109.26) -- cycle ; \draw  [color={rgb, 255:red, 74; green, 144; blue, 226 }  ,draw opacity=1 ][line width=1.5]  (360.85,141.11) .. controls (359.62,141.95) and (358.34,142.73) .. (357,143.44) .. controls (338.12,153.46) and (314.69,146.28) .. (304.68,127.39) .. controls (294.66,108.51) and (301.84,85.09) .. (320.72,75.07) .. controls (320.87,74.99) and (321.02,74.91) .. (321.17,74.84) ;  
\draw  [fill={rgb, 255:red, 0; green, 0; blue, 0 }  ,fill opacity=1 ] (336.2,71) .. controls (336.2,69.62) and (337.32,68.5) .. (338.7,68.5) .. controls (340.08,68.5) and (341.2,69.62) .. (341.2,71) .. controls (341.2,72.38) and (340.08,73.5) .. (338.7,73.5) .. controls (337.32,73.5) and (336.2,72.38) .. (336.2,71) -- cycle ;
\draw  [draw opacity=0][line width=1.5]  (479,113.02) .. controls (479,113.02) and (479,113.02) .. (479,113.02) .. controls (479,78.5) and (496.33,50.52) .. (517.7,50.52) .. controls (526.71,50.52) and (534.99,55.49) .. (541.57,63.82) -- (517.7,113.02) -- cycle ; \draw  [color={rgb, 255:red, 74; green, 144; blue, 226 }  ,draw opacity=1 ][line width=1.5]  (479,113.02) .. controls (479,113.02) and (479,113.02) .. (479,113.02) .. controls (479,78.5) and (496.33,50.52) .. (517.7,50.52) .. controls (526.71,50.52) and (534.99,55.49) .. (541.57,63.82) ;  
\draw  [draw opacity=0][line width=1.5]  (540.41,63.74) .. controls (555.75,63.68) and (568.34,69.04) .. (573.88,79.54) .. controls (583.86,98.45) and (567.2,126.83) .. (536.67,142.94) -- (518.61,108.71) -- cycle ; \draw  [color={rgb, 255:red, 74; green, 144; blue, 226 }  ,draw opacity=1 ][line width=1.5]  (540.41,63.74) .. controls (555.75,63.68) and (568.34,69.04) .. (573.88,79.54) .. controls (583.86,98.45) and (567.2,126.83) .. (536.67,142.94) ;  

\draw (217,8) node [anchor=north west][inner sep=0.75pt]    {$f\in L^{2}( G)$};
\draw (359,33) node [anchor=north west][inner sep=0.75pt]    {$T_{g} f$};
\draw (545,38) node [anchor=north west][inner sep=0.75pt]    {$\mathcal{T} _{\nu _\mathcal{U}} f$};
\draw (207.8,51) node [anchor=north west][inner sep=0.75pt]    {$I$};
\draw (200.8,101) node [anchor=north west][inner sep=0.75pt]    {$G$};
\draw (62.4,101) node [anchor=north west][inner sep=0.75pt]    {$G$};
\draw (69.4,51) node [anchor=north west][inner sep=0.75pt]    {$I$};
\draw (331.9,101) node [anchor=north west][inner sep=0.75pt]    {$G$};
\draw (329.4,77) node [anchor=north west][inner sep=0.75pt]    {$I$};
\draw (511.4,101) node [anchor=north west][inner sep=0.75pt]    {$G$};
\draw (507.4,78) node [anchor=north west][inner sep=0.75pt]    {$I$};
\draw (558.4,80) node [anchor=north west][inner sep=0.75pt]    {$g$};
\draw (378.9,78) node [anchor=north west][inner sep=0.75pt]    {$g$};

\draw (57,160) node [anchor=north west][inner sep=0.75pt]  [font=\large] [align=left] {(a)};
\draw (195,160) node [anchor=north west][inner sep=0.75pt]  [font=\large] [align=left] {(b)};
\draw (330,160) node [anchor=north west][inner sep=0.75pt]  [font=\large] [align=left] {(c)};
\draw (505,160) node [anchor=north west][inner sep=0.75pt]  [font=\large] [align=left] {(d)};
\end{tikzpicture}}
\caption{(a) A group $G$. (b) A function $f \in L^2(G)$ (blue lines). (c) Applying a translation operator $T_g$ on $f$. (d) Applying the average operator $\mathcal{T}_{\nu_{\mathcal{U}}}$ on $f$, where $\mathcal{U}=\{I,g\}$.}
\label{fig:appendix-translation}
\end{figure}
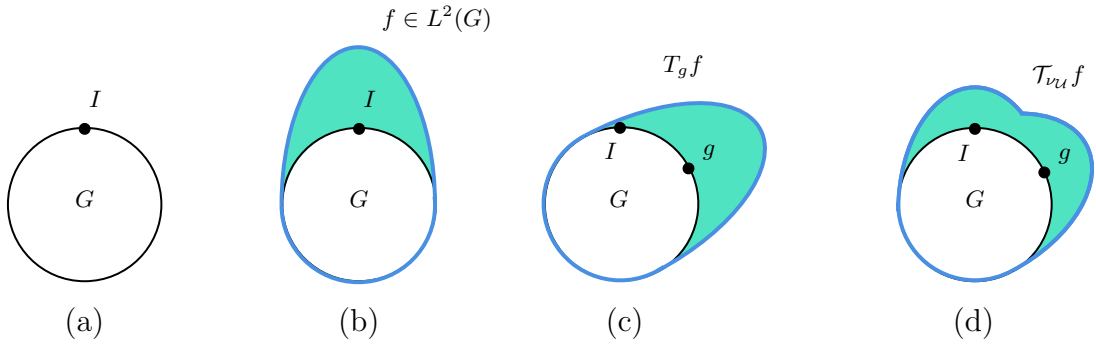
There are several properties of $\mathcal{T}_\nu$:
\begin{enumerate}
\item $\mathcal{T}_\nu$ is self-adjoint, i.e. $\langle \mathcal{T}_\nu f,g \rangle=\langle f,\mathcal{T}_\nu g \rangle$ for any $f,g \in L^2(G)$. Thus the eigenvalues of $T_\nu$ are real.
\item $\mathcal{T}_{\nu_{\mathcal{U}^l}}=\mathcal{T}^l_{\nu_\mathcal{U}}$. Since
\begin{equation*}
\begin{aligned}
    \mathcal{T}^2_{\nu_{\mathcal{U}}}f(h)&=\frac{1}{\lvert \mathcal{U} \rvert^2}\sum_{g_2 \in \mathcal{U}}\sum_{g_1 \in \mathcal{U}}f(g^{-1}_2 g_1^{-1} h)\\
     & =\frac{1}{\lvert \mathcal{U} \rvert^2}\sum_{g \in \mathcal{U}^2}f(g^{-1} h)=\mathcal{T}_{\nu_{\mathcal{U}^2}}f(h),
    \end{aligned}    
\end{equation*}
we have  $\mathcal{T}_{\nu_{\mathcal{U}^2}}=\mathcal{T}^2_{\nu_\mathcal{U}}$ and $\mathcal{T}_{\nu_{\mathcal{U}^l}}=\mathcal{T}^l_{\nu_\mathcal{U}}$.
\item $\mathcal{T}_\mu$ is a projection operator from $L^2(G)$ to its subspace of constant functions, i.e. for any $f \in L^2(G)$, $\mathcal{T}_\mu f=\int_G f(g)d\mu(g)$. And any other $\mathcal{T}_\nu$ acts trivially on this subspace of constant functions. Therefore,
\begin{equation*}
\mathcal{T}_\nu \mathcal{T}_\mu=\mathcal{T}_\mu \mathcal{T}_\nu=\mathcal{T}_\mu.
\end{equation*}
\end{enumerate}

Here we define the spectral gap of $\mathcal{U}$ as
$$
\operatorname{gap}(\mathcal{U})=1-\lVert \mathcal{T}_{\nu_{\mathcal{U}}}-\mathcal{T}_\mu \rVert _{\text{op}},
$$
where the operator norm $\lVert\cdot \rVert _{\text{op}}$ is defined with the 2-norm $\lVert\cdot \rVert _{2}$ on $L^2(G)$
\begin{equation}\label{eq:norm}
\lVert \mathcal{T}_{\nu_{\mathcal{U}}} - \mathcal{T}_\mu \rVert _{\text{op}}=\sup_{\lVert f \rVert_2=1}\lVert (\mathcal{T}_{\nu_{\mathcal{U}}} - \mathcal{T}_\mu) f \rVert_2.
\end{equation}

In linear algebra, the operator norm of a matrix is the largest singular value of the matrix.
It can be shown that the gap is a number in $[0,1]$.
It tries to measure the difference between the average operator of $\nu_\mathcal{U}$ and the average operator of $\mu$. And it also provides an exponentially decayed upper bound for how quickly $\mathcal{T}^l_{\nu_{\mathcal{U}}}$ converges to $\mathcal{T}_{\mu}$ as $l$ goes to infinity. 
\begin{equation}
    \label{eq:a1}
    \begin{aligned}
    \lVert \mathcal{T}^l_{\nu_{\mathcal{U}}}-\mathcal{T}_\mu \rVert _{\text{op}}=\lVert (\mathcal{T}_{\nu_{\mathcal{U}}}-\mathcal{T}_\mu )^l\rVert _{\text{op}} &\leq \lVert \mathcal{T}_{\nu_{\mathcal{U}}}-\mathcal{T}_\mu \rVert _{\text{op}} ^l\\ & =(1-\operatorname{gap}(\mathcal{U}))^l \leq e^{-l\cdot \operatorname{gap}(\mathcal{U})}.
    \end{aligned}  
\end{equation}

The only exception is when the spectral gap is $0$, which can be achieved when a gate set $\mathcal{U}$ is not universal.
For a universal set $\mathcal{U}$, it is believed that its spectral gap is nonzero. Bourgain and Gamburd~\cite{bourgain2012spectral} showed that for $G=\SU(d)$, if the entries of the elements of $\mathcal{U}$ are algebraic numbers, then $\operatorname{gap}(\mathcal{U}) \geq c >0$
for some constant $c$ only depending on $\mathcal{U}$. Yves Benoist and Nicolas de Saxce~\cite{benoist2016spectral} further generalized the result to any compact simple Lie group $G$ and any finite generating Diophantine set $\mathcal{U}$ of $G$.

The indicator function $f$ that tries to maximize the difference between $\mathcal{T}_\nu$ and $\mathcal{T}_\mu$ in \Cref{eq:norm} plays a central role in distinguishing the gate sets with zero and nonzero spectral gaps.
Consider a non-universal set $\mathcal{U}=\{I,H\}$, where $H$ is the Hadamard gate. The indicator function can be chosen to having the following symmetric shape, where the blue curve inside/outside the circle representing group $G$ represents negative/positive function value.

\centerline{
\begin{tikzpicture}[x=0.75pt,y=0.75pt,yscale=-1,xscale=1]

\draw  [color={rgb, 255:red, 74; green, 144; blue, 226 }  ,draw opacity=1 ][fill={rgb, 255:red, 80; green, 227; blue, 194 }  ,fill opacity=1 ][line width=1.5]  (430.25,101.7) .. controls (430.25,60.37) and (441.64,26.86) .. (455.7,26.86) .. controls (469.76,26.86) and (481.15,60.37) .. (481.15,101.7) .. controls (481.15,143.03) and (469.76,176.54) .. (455.7,176.54) .. controls (441.64,176.54) and (430.25,143.03) .. (430.25,101.7) -- cycle ;
\draw  [fill={rgb, 255:red, 255; green, 255; blue, 255 }  ,fill opacity=1 ] (417,101.7) .. controls (417,80.33) and (434.33,63) .. (455.7,63) .. controls (477.07,63) and (494.4,80.33) .. (494.4,101.7) .. controls (494.4,123.07) and (477.07,140.4) .. (455.7,140.4) .. controls (434.33,140.4) and (417,123.07) .. (417,101.7) -- cycle ;
\draw  [draw opacity=0][fill={rgb, 255:red, 80; green, 227; blue, 194 }  ,fill opacity=1 ] (494.4,101.7) .. controls (494.4,114.38) and (488.3,125.64) .. (478.87,132.7) .. controls (480.33,123.25) and (481.15,112.76) .. (481.15,101.7) .. controls (481.15,90.64) and (480.33,80.15) .. (478.87,70.7) .. controls (488.3,77.76) and (494.4,89.02) .. (494.4,101.7) -- cycle (417,101.7) .. controls (417,89.02) and (423.1,77.76) .. (432.53,70.7) .. controls (431.07,80.15) and (430.25,90.64) .. (430.25,101.7) .. controls (430.25,112.76) and (431.07,123.25) .. (432.53,132.7) .. controls (423.1,125.64) and (417,114.38) .. (417,101.7) -- cycle ;
\draw  [fill={rgb, 255:red, 0; green, 0; blue, 0 }  ,fill opacity=1 ] (453.2,63.5) .. controls (453.2,62.12) and (454.32,61) .. (455.7,61) .. controls (457.08,61) and (458.2,62.12) .. (458.2,63.5) .. controls (458.2,64.88) and (457.08,66) .. (455.7,66) .. controls (454.32,66) and (453.2,64.88) .. (453.2,63.5) -- cycle ;
\draw   (277,96.7) .. controls (277,75.33) and (294.33,58) .. (315.7,58) .. controls (337.07,58) and (354.4,75.33) .. (354.4,96.7) .. controls (354.4,118.07) and (337.07,135.4) .. (315.7,135.4) .. controls (294.33,135.4) and (277,118.07) .. (277,96.7) -- cycle ;
\draw  [fill={rgb, 255:red, 0; green, 0; blue, 0 }  ,fill opacity=1 ] (313.2,58.5) .. controls (313.2,57.12) and (314.32,56) .. (315.7,56) .. controls (317.08,56) and (318.2,57.12) .. (318.2,58.5) .. controls (318.2,59.88) and (317.08,61) .. (315.7,61) .. controls (314.32,61) and (313.2,59.88) .. (313.2,58.5) -- cycle ;
\draw  [fill={rgb, 255:red, 0; green, 0; blue, 0 }  ,fill opacity=1 ] (313.2,134.9) .. controls (313.2,133.52) and (314.32,132.4) .. (315.7,132.4) .. controls (317.08,132.4) and (318.2,133.52) .. (318.2,134.9) .. controls (318.2,136.28) and (317.08,137.4) .. (315.7,137.4) .. controls (314.32,137.4) and (313.2,136.28) .. (313.2,134.9) -- cycle ;
\draw  [fill={rgb, 255:red, 0; green, 0; blue, 0 }  ,fill opacity=1 ] (453.2,140.9) .. controls (453.2,139.52) and (454.32,138.4) .. (455.7,138.4) .. controls (457.08,138.4) and (458.2,139.52) .. (458.2,140.9) .. controls (458.2,142.28) and (457.08,143.4) .. (455.7,143.4) .. controls (454.32,143.4) and (453.2,142.28) .. (453.2,140.9) -- cycle ;
\draw   (417,101.7) .. controls (417,80.33) and (434.33,63) .. (455.7,63) .. controls (477.07,63) and (494.4,80.33) .. (494.4,101.7) .. controls (494.4,123.07) and (477.07,140.4) .. (455.7,140.4) .. controls (434.33,140.4) and (417,123.07) .. (417,101.7) -- cycle ;
\draw  [color={rgb, 255:red, 74; green, 144; blue, 226 }  ,draw opacity=1 ][line width=1.5]  (430.25,101.7) .. controls (430.25,60.37) and (441.64,26.86) .. (455.7,26.86) .. controls (469.76,26.86) and (481.15,60.37) .. (481.15,101.7) .. controls (481.15,143.03) and (469.76,176.54) .. (455.7,176.54) .. controls (441.64,176.54) and (430.25,143.03) .. (430.25,101.7) -- cycle ;

\draw (185,86) node [anchor=north west][inner sep=0.75pt]    {$\mathcal{U}=\{I,H\}$};
\draw (309.4,88) node [anchor=north west][inner sep=0.75pt]    {$G$};
\draw (316.4,38) node [anchor=north west][inner sep=0.75pt]    {$I$};
\draw (320.2,137.9) node [anchor=north west][inner sep=0.75pt]    {$H$};
\draw (457.7,143.9) node [anchor=north west][inner sep=0.75pt]    {$H$};
\draw (520,86) node [anchor=north west][inner sep=0.75pt]    {$\mathcal{T}_{\nu _\mathcal{U}} f=T_{H} f=f$};
\draw (427,16) node [anchor=north west][inner sep=0.75pt]    {$f$};
\draw (449.4,93) node [anchor=north west][inner sep=0.75pt]    {$G$};
\draw (456.4,43) node [anchor=north west][inner sep=0.75pt]    {$I$};

\end{tikzpicture}}

We further require this function to have mean value $0$ on measure $\mu$. Then we can show the spectral gap of $\mathcal{U}$ is also $0$. The average operator $\mathcal{T}_{\nu_{\mathcal{U}}}$ acts trivially on this function, which indicates $
\lVert \mathcal{T}_{\nu_{\mathcal{U}}} f - \mathcal{T}_\mu f \rVert_2 =\lVert \mathcal{T}_{\nu_{\mathcal{U}}} f \rVert_2 = \lVert f \rVert_2$ = 1.

\subsection{Spectral Gaps and $\epsilon$-Nets}
We will show that the exponentially decaying upper bound of the spectral gap in \Cref{eq:a1} can be used to bound the compiling precision $\epsilon$.
To characterize the compiling precision, we adopt the concept of $\epsilon$-net.
An $\epsilon$-net of $G$ is a finite subset $E$ of $G$ such that for any $g \in G$, there exists $g_0 \in E$ such that $d(g,g_0) \leq \epsilon$.

\begin{theorem}~\cite{slowik2023calculable}
Let $\mathcal{U}$ be a finite generating set of group $G$, and $\operatorname{gap}(\mathcal{U})>0$. For $\epsilon>0$, we have $\mathcal{U}^l$ being an $\epsilon$-net for any $l$ that satisfies
\begin{equation}\label{eq:gap-depth}
l > \frac{\dim(G)}{\operatorname{gap}(\mathcal{U})}\log (1/\epsilon)+B,
\end{equation}
for some constant $B$.
\label{thm:gap}
\end{theorem}

\begin{proof} 
We proof the theorem by contradiction.
Equation (\ref{eq:a1}) shows that for any $\mathcal{U}$ with nonzero gap, $\mathcal{T}^l_{\nu_{\mathcal{U}}}$ converges to $\mathcal{T}_{\mu}$ as $l$ goes to infinity. In the following, we show if $\mathcal{U}^l$ is not an $\epsilon$-net of $G$, there exists a function to effectively indicate the difference between $\mathcal{T}^l_{\nu_{\mathcal{U}}}$ and $\mathcal{T}_{\mu}$, which violates \Cref{eq:a1} and gives us a contradiction.
Under the assumption that $\mathcal{U}^l$ is not an $\epsilon$-net of $G$, there exists a $g^* \in G$ such that for any $w_l \in \mathcal{U}^l$, $d(g^*,w_l)\geq \epsilon$.
We can define two $\epsilon/2$ balls $\Omega = B_{\epsilon/2}(g^*)$ and $\Omega_0 =B_{\epsilon/2}(I)$ centered at $g^*$ and $I$ respectively as shown in the following figure.

\tikzset{every picture/.style={line width=0.75pt}} 
\centerline{\begin{tikzpicture}[x=0.75pt,y=0.75pt,yscale=-0.8,xscale=0.8]
    \draw  [fill={rgb, 255:red, 248; green, 231; blue, 28 }  ,fill opacity=1 ] (263.3,220) .. controls (263.3,195.87) and (282.87,176.3) .. (307,176.3) .. controls (331.13,176.3) and (350.7,195.87) .. (350.7,220) .. controls (350.7,244.13) and (331.13,263.7) .. (307,263.7) .. controls (282.87,263.7) and (263.3,244.13) .. (263.3,220) -- cycle ;
    \draw   (342.69,141.7) .. controls (342.69,99.09) and (377.24,64.54) .. (419.85,64.54) .. controls (462.46,64.54) and (497.01,99.09) .. (497.01,141.7) .. controls (497.01,184.31) and (462.46,218.86) .. (419.85,218.86) .. controls (377.24,218.86) and (342.69,184.31) .. (342.69,141.7) -- cycle ;
    \draw   (219.4,163.7) .. controls (219.4,96.32) and (286.33,41.7) .. (368.9,41.7) .. controls (451.47,41.7) and (518.4,96.32) .. (518.4,163.7) .. controls (518.4,231.08) and (451.47,285.7) .. (368.9,285.7) .. controls (286.33,285.7) and (219.4,231.08) .. (219.4,163.7) -- cycle ;
    \draw  [fill={rgb, 255:red, 0; green, 0; blue, 0 }  ,fill opacity=1 ] (304.15,220) .. controls (304.15,218.43) and (305.43,217.15) .. (307,217.15) .. controls (308.57,217.15) and (309.85,218.43) .. (309.85,220) .. controls (309.85,221.57) and (308.57,222.85) .. (307,222.85) .. controls (305.43,222.85) and (304.15,221.57) .. (304.15,220) -- cycle ;
    \draw  [fill={rgb, 255:red, 0; green, 0; blue, 0 }  ,fill opacity=1 ] (417,141.7) .. controls (417,140.13) and (418.28,138.85) .. (419.85,138.85) .. controls (421.42,138.85) and (422.7,140.13) .. (422.7,141.7) .. controls (422.7,143.27) and (421.42,144.55) .. (419.85,144.55) .. controls (418.28,144.55) and (417,143.27) .. (417,141.7) -- cycle ;
    \draw    (307,220) -- (343.4,195.5) ;
    \draw  [fill={rgb, 255:red, 248; green, 231; blue, 28 }  ,fill opacity=1 ] (376.15,141.7) .. controls (376.15,117.57) and (395.72,98) .. (419.85,98) .. controls (443.98,98) and (463.55,117.57) .. (463.55,141.7) .. controls (463.55,165.83) and (443.98,185.4) .. (419.85,185.4) .. controls (395.72,185.4) and (376.15,165.83) .. (376.15,141.7) -- cycle ;
    \draw    (419.85,141.7) -- (361.8,91.07) ;
    \draw    (455.8,117.07) -- (419.85,141.7) ;
    \draw  [fill={rgb, 255:red, 0; green, 0; blue, 0 }  ,fill opacity=1 ] (417,141.7) .. controls (417,140.13) and (418.28,138.85) .. (419.85,138.85) .. controls (421.42,138.85) and (422.7,140.13) .. (422.7,141.7) .. controls (422.7,143.27) and (421.42,144.55) .. (419.85,144.55) .. controls (418.28,144.55) and (417,143.27) .. (417,141.7) -- cycle ;
    
    \draw (292,219) node [anchor=north west][inner sep=0.75pt]    {$I$};
    \draw (304.2,191) node [anchor=north west][inner sep=0.75pt]    {$\epsilon /2$};
    \draw (273,80) node [anchor=north west][inner sep=0.75pt]    {$G$};
    \draw (431.53,129.67) node [anchor=north west][inner sep=0.75pt]    {$\epsilon /2$};
    \draw (387.53,97.67) node [anchor=north west][inner sep=0.75pt]    {$\epsilon $};
    \draw (336.67,251.67) node [anchor=north west][inner sep=0.75pt]    {$\Omega _{0}$};
    \draw (443.33,177) node [anchor=north west][inner sep=0.75pt]    {$\Omega $};
    \draw (408,145.33) node [anchor=north west][inner sep=0.75pt]    {$g^*$};
\end{tikzpicture}}
    
The indicator function $f$ on $G$ can be defined as $f(g)=\frac{1}{V(\Omega_0)}\chi_{\Omega_0}(g)$, where
$$
\chi_{\Omega}(g)=
\left\{
    \begin{aligned}
   1, &\hspace*{0.3cm} g\in \Omega;\\
    0, & \hspace*{0.3cm} g \notin \Omega,
    \end{aligned}
\right.
$$
and $V(\Omega_0)$ is the volume of $\Omega_0$. The function $f$ is a constant function on $\Omega_0$ and zero elsewhere. It has $1$-norm $\lVert f \rVert_1=1$ and $2$-norm $\lVert f \rVert_2=1/\sqrt{V(\Omega_0)}$.
In the following, we will show this indicator function can be used to quantify the difference between $\mathcal{T}^l_{\nu_{\mathcal{U}}}$ and $\mathcal{T}_{\mu}$, which gives us a lower bound of $\|\mathcal{T}_{\nu_{\mathcal{U}}}^l-\mathcal{T}_{\mu}\|_{\text{op}}$ that appeared in \Cref{eq:a1}.

\begin{equation*}
    \left\Vert \mathcal{T}_{\nu _{\mathcal{\mathcal{U}}} }^{l } f- \mathcal{T}_{\mu } f\right\Vert _{2}\leq \left\Vert  \mathcal{T}_{\nu _{\mathcal{\mathcal{U}}}}^{l } -\mathcal{T}_{\mu }\right\Vert _{\mathrm{op}} \cdot \lVert f\rVert _{2}  \leq e^{-l \cdot \operatorname{gap} (\mathcal{U})}  \frac{1}{\sqrt{V(\Omega_0)}}.
\end{equation*}

Therefore, we have
\begin{equation}\label{eq:a2}
\sqrt{V(\Omega_0)} \cdot \left\Vert  \mathcal{T}_{\nu _{\mathcal{\mathcal{U}}}}^{l } f-\mathcal{T}_{\mu }f\right\Vert _{2} \leq e^{-l \cdot \operatorname{gap} (\mathcal{U})}.
\end{equation}

On the other hand, the left side of \Cref{eq:a2} can be calculated as
\begin{equation}\label{eq:a3}
    \begin{aligned}
    &\ \ \ \ \sqrt{V(\Omega_0)} \cdot \left\Vert  \mathcal{T}_{\nu _{\mathcal{\mathcal{U}}}}^{l } f-\mathcal{T}_{\mu }f\right\Vert _{2} \\  
    &=  \sqrt{V(\Omega)} \cdot \left\Vert  \mathcal{T}_{\nu _{\mathcal{\mathcal{U}}}}^{l } f-\mathcal{T}_{\mu }f\right\Vert _{2} \\ 
    &=  \lVert \chi _{\Omega }\rVert _{2}\cdot \left\Vert \mathcal{T}_{\mu } f-\mathcal{T}_{\nu _{\mathcal{\mathcal{U}}} }^{l } f\right\Vert _{2}\\
    &\geq \left< \chi _{\Omega } \mid \mathcal{T}_{\mu } f-\mathcal{T}_{\nu _{\mathcal{\mathcal{U}}} }^{l } f \right> \\
    & =\int _{G }\chi _{\Omega } \left(  \mathcal{T}_{\mu } f -\mathcal{T}_{\nu _{\mathcal{\mathcal{U}}}}^{l } f\right)d \mu \\ 
    & =\int _{\Omega } \left(\mathcal{T}_{\mu } f- \mathcal{T}_{\nu _{\mathcal{\mathcal{U}}}}^{l } f\right) d \mu \\
    &= \int _{\Omega }\mathcal{T}_{\mu } fd \mu -\int _{\Omega } \mathcal{T}_{\nu _{\mathcal{\mathcal{U}}}}^{l } fd \mu.
    \end{aligned}
    \end{equation}
The reason why $f$ effectively indicate the difference between $\mathcal{T}^l_{\nu_{\mathcal{U}}}$ and $\mathcal{T}_{\mu}$ is that $\mathcal{T}_{\mu } f$ is a constant function and $\mathcal{T}_{\nu _{\mathcal{\mathcal{U}}}}^{l } f$ is zero on $\Omega$. We will show this in the following. For the Haar measure $\mu$, we have
\begin{equation}
    \label{eq:a4}
    \int _{\Omega } \mathcal{T}_{\mu } fd \mu =\int _{\Omega }1d \mu =V(\Omega).
\end{equation}
With the definition for the discrete measure in \Cref{eq:translation-average}, we have
$$
\int _{\Omega }\left( \mathcal{T}_{\nu _{\mathcal{\mathcal{U}}}}^{l } f\right) (g)d \mu (g)=\frac{1}{\lvert \mathcal{U} \rvert ^l}\sum _{w_l\in \mathcal{U}^l}\int _{\Omega } f({w_l}^{-1} g)d \mu (g).
$$
For all $g \in \Omega$, we have $d(w_l^{-1}g,I)=d(g,w_l)>\epsilon/2$ (the figure above). Thus,
$$
f({w_l}^{-1} g)=\frac{1}{V(\Omega_0)}\chi_{\Omega_0}({w_l}^{-1} g)=0.
$$
Hence,
$$
\int _{\Omega } \mathcal{T}_{\nu _{\mathcal{\mathcal{U}}}}^{l } f d \mu =0.
$$
Combining this with \Cref{eq:a2}, \Cref{eq:a3} and \Cref{eq:a4}, we have
\begin{equation}
    \label{eq:a5}
    V(\Omega) \leq e^{-l \cdot \operatorname{gap} (\mathcal{U})}.
\end{equation}
Recall that $V(\Omega)=C_V(\epsilon/2)^{\dim(G)}$ with some constant number $C_V$, \Cref{eq:a5} is equivalent to
$$
l \leq \frac{\dim(G)}{\operatorname{gap}(\mathcal{U})}\log (1/\epsilon)+B,
$$
where  $B=-\frac{\log(C_V)-\dim(G)\cdot\log2}{\operatorname{gap} (\mathcal{U})}$ is a constant that only depends on $\mathcal{U}$ and $G$.
That contradicts with \Cref{eq:gap-depth}. Therefore, $\mathcal{U}^l$ is an $\epsilon$-net of $G$.
\end{proof}

\section{Inequality between Frobenius Norm and Infidelity}
\label{sec:appendix-infvsfro}
In this section, we show that the Frobenius norm of the difference between two unitary matrices is an upper bound to the square root of the infidelity between them. Let $\displaystyle A=( a_{i,j})$ and $\displaystyle B=( b_{i,j})$ be two $\displaystyle N\times N$ unitary matrices. We have  
\begin{equation*}
\begin{aligned}
    \epsilon_{I}(A,B) & =1-\frac{\left| \operatorname{Tr}\left( A^{\dagger } B\right)\right| }{N}\\
 & =\frac{1}{2N}\left( 2N-\left| \sum\limits _{i,j}\overline{a}_{i,j} b_{i,j}\right| -\left| \sum\limits _{i,j}\overline{b}_{i,j} a_{i,j}\right| \right)\\
 & =\frac{1}{2N}\left(\sum\limits _{i,j}\overline{a}_{i,j} a_{i,j} +\sum\limits _{i,j}\overline{b}_{i,j} b_{i,j} -\left| \sum\limits _{i,j}\overline{a}_{i,j} b_{i,j}\right| -\left| \sum\limits _{i,j}\overline{b}_{i,j} a_{i,j}\right| \right)\\
 & \leq \frac{1}{2N}\left(\sum\limits _{i,j}\overline{a}_{i,j} a_{i,j} +\sum\limits _{i,j}\overline{b}_{i,j} b_{i,j} -\sum\limits _{i,j}\overline{a}_{i,j} b_{i,j} -\sum\limits _{i,j}\overline{b}_{i,j} a_{i,j}\right)\\
 & =\frac{1}{2N}\sum\limits _{i,j}(\overline{a}_{i,j} a_{i,j} +\overline{b}_{i,j} b_{i,j} -\overline{a}_{i,j} b_{i,j} -\overline{b}_{i,j} a_{i,j})\\
 & =\frac{1}{2N}\sum\limits _{i,j}\left| a_{i,j} -b_{i,j}\right| ^{2} =\frac{1}{2N} \epsilon^2_F(A,B).
\end{aligned}
\end{equation*}

\section{Compilation Algorithm}
\label{sec:appendix-mim}

An exhaustive search over a depth $d$ circuit for a target unitary $\Gamma$ requires
$4^{d}$ number of unitaries to be computed via matrix multiplication. Half of such computation can be avoided via an improved brute-force method~\cite{amy2013algorithms}, shown in
\ref{alg:double}. \LinesNumberedHidden
\begin{algorithm}
    \small
    \SetAlgoNoLine
Let $\mathcal{S}$ be the set of $4^d$ gates implemented by a circuit with depth $d$, $\Gamma$ be the target gate\;
$\Gamma_{2d} \gets I$ \;
\ForEach{$X$ in $\mathcal{S}$}{
    $L \gets \Gamma X^{-1}$\;
    $\tilde{L} \gets \texttt{find\_closest}(\mathcal{S}, L)$ \tcp*[r]{Costly query over $4^d$ gates.}
    \If{$\texttt{dist}(\Gamma, \tilde{L}X) < \texttt{dist}(\Gamma, \Gamma_{2d})$}{
        $\Gamma_{2d} \gets \tilde{L}X$ \tcp*[r]{Best approximation of $\Gamma$.}
    }
}\caption{Meet-in-the-middle Algorithm for Unitary Compilation}\label{alg:double}
\end{algorithm}
We are able to use this algorithm to effectively square the precision of
compilation with cost of $\mathcal{O}(d^{m})$ with $m > 0$. The precision is characterized by operator infidelity, an approximate metric. Therefore, a data structure could be constrcuted to expedite the $\texttt{find\_closest}$ which finds the closest operator to $\tilde{L}$ amongst $4^d$ operators. One such data structure is the KD-tree~\cite{bentley1975multidimensional}. A
KD-tree is a tree where tree nodes each stores a condition that partitions the feature
space of data into two parts and the tree leaves store the data that satisfy all
conditions from root to current leaf node. Each leaf node contains points that are close to each other in the feature space. Hence, KD-tree is able to reduce the search cost to $\mathcal{O}(\log n)$ in the best case~\cite{bentley1975multidimensional}. However, in high-dimensional space where data are distributed uniformly, closeness of points loses it meaning. All points are as far away from the target as any other. This is the case for our compiling problem where the unitary gates reside in the $16$-dimensional Euclidean space. As a result, performance of querying drops to $\mathcal{O}(n)$ as was observed in the numerical experiment. As a last resort, we employ an approximate nearest neighbor algorithm, Hierarchical Navigable Small World algorithm~\cite{tellez2022similaritysearch,malkov2018efficient}, for $\mathcal{O}(\log{n})$ search of the approxmiately best compiling gate set. The Hierarchical Navigable Small World algorithm is an graph-based approximate nearest neighbor query algorithm~\cite{malkov2018efficient}. It combines the concept of skipped linked list to allow for enlarged neighborhood search, and the concept of small world network to approximate a Delaunday triangulation to efficiently search for nearest neighbors.  

\section{Compilation for gate set with only $U$ and $U^S$}
\label{sec:appendix-twou}
The results of compilation for gate set $\mathcal{U}= \{ U,U^S \}$ are shown in \Cref{fig:twou}. We observe a similar scaling of $d \sim \log(1/\epsilon)$ with a prefactor $A^*_I=\frac{N^2-1}{2\log_{10}|\mathcal{U}|}\approx 24.92$.
\begin{figure}
    \includegraphics[width=0.5\columnwidth]{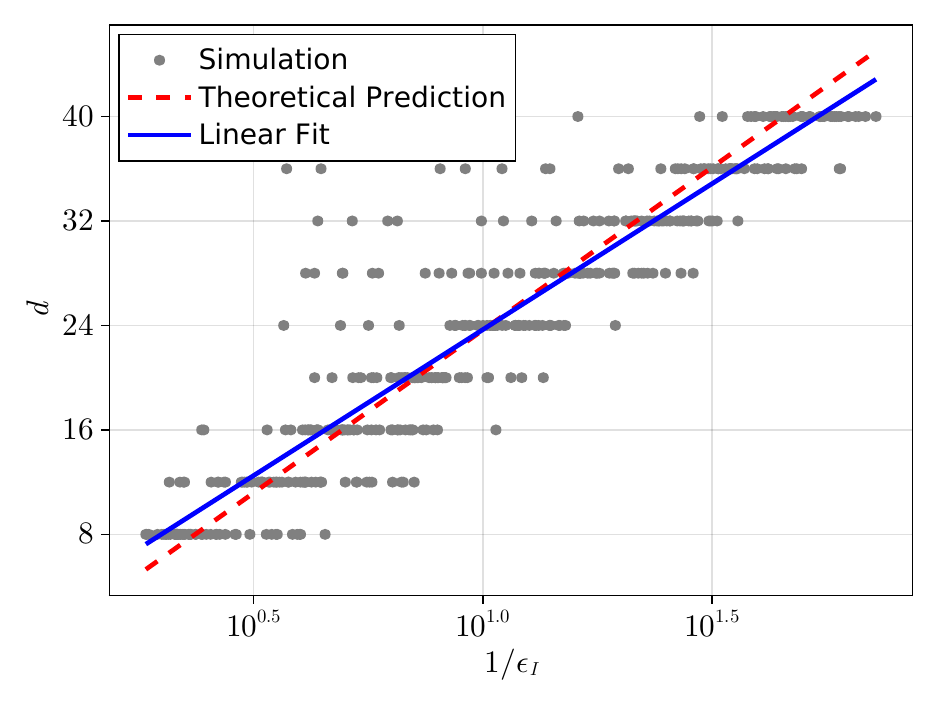}
    \caption{The required depth $d$ for compiling $CNOT$ gate with only one arbitrary $U$ and its variant $U^S$ to precision $\epsilon_I$. Each gray dot represents a randomly chosen $U$, $50$ different instances are presented. The blue line is the linear fitting of the experimental data. The red dashed line is the theoretical estimation.}
    \label{fig:twou}
\end{figure}

\section{Measurement-Free Quantum Error Correction for Repetition Code and Shor Code}
\label{sec:appendix-mfqec}
In this section, we use the 3-qubit bit-flip repetition code as an example to illustrate the MF-QEC protocol and show the MF-QEC circuit for the Shor code. The 3-qubit bit-flip repetition code is defined by the stabilizers $Z_1Z_2$ and $Z_2Z_3$. The logical state $|\overline{0}\rangle _{b} $ and $|\overline{1}\rangle _{b}$ are
$$
\begin{aligned}
    |\overline{0}\rangle _{b}&= |000\rangle \\
    |\overline{1}\rangle _{b}&= |111\rangle.
\end{aligned}
$$
Besides $Z_1Z_2$ and $Z_2Z_3$, the stabilizer group also includes $Z_1Z_3$. Each bit-flip error on a single qubit will cause two of the stabilizers to have the eigenvalue of $-1$, shown in \Cref{fig:reptcode}. And, these stabilizers will indicate the location of the error. We use them to control the correction operators. The circuit of syndrome extraction and controlled error correction for repetition code is shown in \Cref{fig:qcbf}.
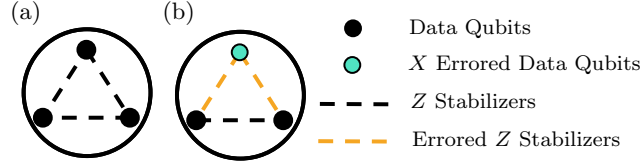
\begin{figure}

    \tikzset{every picture/.style={line width=0.75pt}} 

    \begin{tikzpicture}[x=0.75pt,y=0.75pt,yscale=-0.8,xscale=0.8]
    
    \draw [line width=1.5]  [dash pattern={on 5.63pt off 4.5pt}]  (201.17,73.5) -- (247.33,73.75) ;
    \draw [color={rgb, 255:red, 245; green, 166; blue, 35 }  ,draw opacity=1 ][line width=1.5]  [dash pattern={on 5.63pt off 4.5pt}]  (200.17,97) -- (246.33,97.25) ;
    \draw [color={rgb, 255:red, 245; green, 166; blue, 35 }  ,draw opacity=1 ][line width=1.5]  [dash pattern={on 5.63pt off 4.5pt}]  (150.63,40.93) -- (178.07,83.8) ;
    \draw [color={rgb, 255:red, 245; green, 166; blue, 35 }  ,draw opacity=1 ][line width=1.5]  [dash pattern={on 5.63pt off 4.5pt}]  (150.63,40.93) -- (123.19,83.8) ;
    \draw  [line width=1.5]  (111.1,68) .. controls (111.1,45.91) and (129.01,28) .. (151.1,28) .. controls (173.2,28) and (191.1,45.91) .. (191.1,68) .. controls (191.1,90.1) and (173.2,108) .. (151.1,108) .. controls (129.01,108) and (111.1,90.1) .. (111.1,68) -- cycle ;
    \draw  [fill={rgb, 255:red, 0; green, 0; blue, 0 }  ,fill opacity=1 ][line width=1.5]  (118.19,83.8) .. controls (118.19,81.04) and (120.43,78.8) .. (123.19,78.8) .. controls (125.95,78.8) and (128.19,81.04) .. (128.19,83.8) .. controls (128.19,86.57) and (125.95,88.8) .. (123.19,88.8) .. controls (120.43,88.8) and (118.19,86.57) .. (118.19,83.8) -- cycle ;
    \draw  [fill={rgb, 255:red, 0; green, 0; blue, 0 }  ,fill opacity=1 ][line width=1.5]  (173.07,83.8) .. controls (173.07,81.04) and (175.31,78.8) .. (178.07,78.8) .. controls (180.83,78.8) and (183.07,81.04) .. (183.07,83.8) .. controls (183.07,86.57) and (180.83,88.8) .. (178.07,88.8) .. controls (175.31,88.8) and (173.07,86.57) .. (173.07,83.8) -- cycle ;
    \draw  [line width=1.5]  (111.07,68) .. controls (111.07,45.91) and (128.98,28) .. (151.07,28) .. controls (173.16,28) and (191.07,45.91) .. (191.07,68) .. controls (191.07,90.1) and (173.16,108) .. (151.07,108) .. controls (128.98,108) and (111.07,90.1) .. (111.07,68) -- cycle ;
    \draw  [fill={rgb, 255:red, 80; green, 227; blue, 194 }  ,fill opacity=1 ] (145.63,40.93) .. controls (145.63,38.17) and (147.87,35.93) .. (150.63,35.93) .. controls (153.39,35.93) and (155.63,38.17) .. (155.63,40.93) .. controls (155.63,43.69) and (153.39,45.93) .. (150.63,45.93) .. controls (147.87,45.93) and (145.63,43.69) .. (145.63,40.93) -- cycle ;
    \draw [color={rgb, 255:red, 0; green, 0; blue, 0 }  ,draw opacity=1 ][line width=1.5]  [dash pattern={on 5.63pt off 4.5pt}]  (123.19,83.8) -- (178.07,83.8) ;
    
    \draw  [line width=1.5]  (14.44,66.4) .. controls (14.44,44.31) and (32.35,26.4) .. (54.44,26.4) .. controls (76.53,26.4) and (94.44,44.31) .. (94.44,66.4) .. controls (94.44,88.49) and (76.53,106.4) .. (54.44,106.4) .. controls (32.35,106.4) and (14.44,88.49) .. (14.44,66.4) -- cycle ;
    \draw  [fill={rgb, 255:red, 0; green, 0; blue, 0 }  ,fill opacity=1 ][line width=1.5]  (21.52,82.2) .. controls (21.52,79.44) and (23.76,77.2) .. (26.52,77.2) .. controls (29.29,77.2) and (31.52,79.44) .. (31.52,82.2) .. controls (31.52,84.96) and (29.29,87.2) .. (26.52,87.2) .. controls (23.76,87.2) and (21.52,84.96) .. (21.52,82.2) -- cycle ;
    \draw  [fill={rgb, 255:red, 0; green, 0; blue, 0 }  ,fill opacity=1 ][line width=1.5]  (76.4,82.2) .. controls (76.4,79.44) and (78.64,77.2) .. (81.4,77.2) .. controls (84.16,77.2) and (86.4,79.44) .. (86.4,82.2) .. controls (86.4,84.96) and (84.16,87.2) .. (81.4,87.2) .. controls (78.64,87.2) and (76.4,84.96) .. (76.4,82.2) -- cycle ;
    \draw  [fill={rgb, 255:red, 0; green, 0; blue, 0 }  ,fill opacity=1 ][line width=1.5]  (48.96,39.32) .. controls (48.96,36.56) and (51.2,34.32) .. (53.96,34.32) .. controls (56.72,34.32) and (58.96,36.56) .. (58.96,39.32) .. controls (58.96,42.09) and (56.72,44.32) .. (53.96,44.32) .. controls (51.2,44.32) and (48.96,42.09) .. (48.96,39.32) -- cycle ;
    \draw  [dash pattern={on 5.63pt off 4.5pt}][line width=1.5]  (53.96,39.32) -- (81.4,82.2) -- (26.52,82.2) -- cycle ;
    \draw  [line width=1.5]  (14.4,66.4) .. controls (14.4,44.31) and (32.31,26.4) .. (54.4,26.4) .. controls (76.49,26.4) and (94.4,44.31) .. (94.4,66.4) .. controls (94.4,88.49) and (76.49,106.4) .. (54.4,106.4) .. controls (32.31,106.4) and (14.4,88.49) .. (14.4,66.4) -- cycle ;
    
    \draw  [fill={rgb, 255:red, 80; green, 227; blue, 194 }  ,fill opacity=1 ] (216.63,48.93) .. controls (216.63,46.17) and (218.87,43.93) .. (221.63,43.93) .. controls (224.39,43.93) and (226.63,46.17) .. (226.63,48.93) .. controls (226.63,51.69) and (224.39,53.93) .. (221.63,53.93) .. controls (218.87,53.93) and (216.63,51.69) .. (216.63,48.93) -- cycle ;
    \draw  [fill={rgb, 255:red, 0; green, 0; blue, 0 }  ,fill opacity=1 ] (217.07,25.8) .. controls (217.07,23.04) and (219.31,20.8) .. (222.07,20.8) .. controls (224.83,20.8) and (227.07,23.04) .. (227.07,25.8) .. controls (227.07,28.57) and (224.83,30.8) .. (222.07,30.8) .. controls (219.31,30.8) and (217.07,28.57) .. (217.07,25.8) -- cycle ;
    
    \draw (255.5,18.5) node [anchor=north west][inner sep=0.75pt]  [font=\footnotesize] [align=left] {{\fontfamily{helvet}\selectfont Data Qubits}};
    \draw (256.5,65.17) node [anchor=north west][inner sep=0.75pt]  [font=\footnotesize] [align=left] {$\displaystyle Z$ {Stabilizers}};
    \draw (255,41) node [anchor=north west][inner sep=0.75pt]  [font=\footnotesize] [align=left] {$\displaystyle X$ {Errored Data Qubits}};
    \draw (257.5,88.67) node [anchor=north west][inner sep=0.75pt]  [font=\footnotesize] [align=left] {{Errored }$\displaystyle Z$ {Stabilizers}};
    \draw (4,7) node [anchor=north west][inner sep=0.75pt]   [align=left] {(a)};
    \draw (100,7) node [anchor=north west][inner sep=0.75pt]   [align=left] {(b)};
    \end{tikzpicture}
\caption{3-qubit bit-flip repetition code. (a) The dashed lines between qubits represent stabilizers $Z_1Z_2$, $Z_2Z_3$ and $Z_1Z_3$. (b) When there is an $X$ error on one qubit, the $Z$ stabilizers near the error qubit will become syndrome.}\label{fig:reptcode}
\end{figure}

Shor code is the concatenation of two repetition code. Therefore, we can use a similar MF-QEC method for Shor code~\cite{veroni2024optimized}, shown in \Cref{fig:qcshorx} and \Cref{fig:qcshorz}.

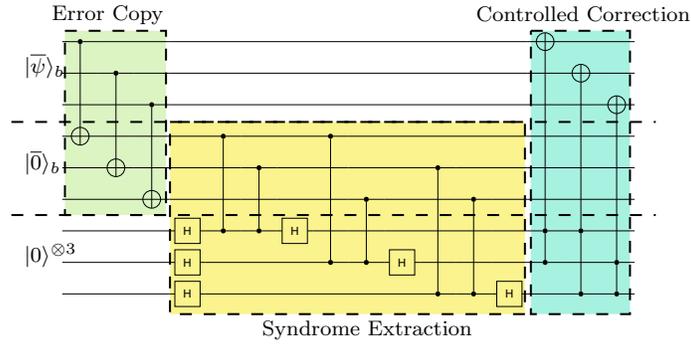
\begin{figure}
    \tikzset{every picture/.style={line width=0.75pt}} 

    \begin{tikzpicture}[x=0.75pt,y=0.75pt,yscale=-1,xscale=1]
    
    \draw  [fill={rgb, 255:red, 248; green, 231; blue, 28 }  ,fill opacity=0.5 ][dash pattern={on 4.5pt off 4.5pt}] (85,66) -- (264,66) -- (264,163) -- (85,163) -- cycle ;

    \draw  [fill={rgb, 255:red, 184; green, 233; blue, 134 }  ,fill opacity=0.5 ][dash pattern={on 4.5pt off 4.5pt}] (32,20) -- (83,20) -- (83,113) -- (32,113) -- cycle ;

    \draw  [fill={rgb, 255:red, 80; green, 227; blue, 194 }  ,fill opacity=0.5 ][dash pattern={on 4.5pt off 4.5pt}] (267,20) -- (317,20) -- (317,163) -- (267,163) -- cycle ;

    \draw (175,89.14) node  {\includegraphics[width=230pt]{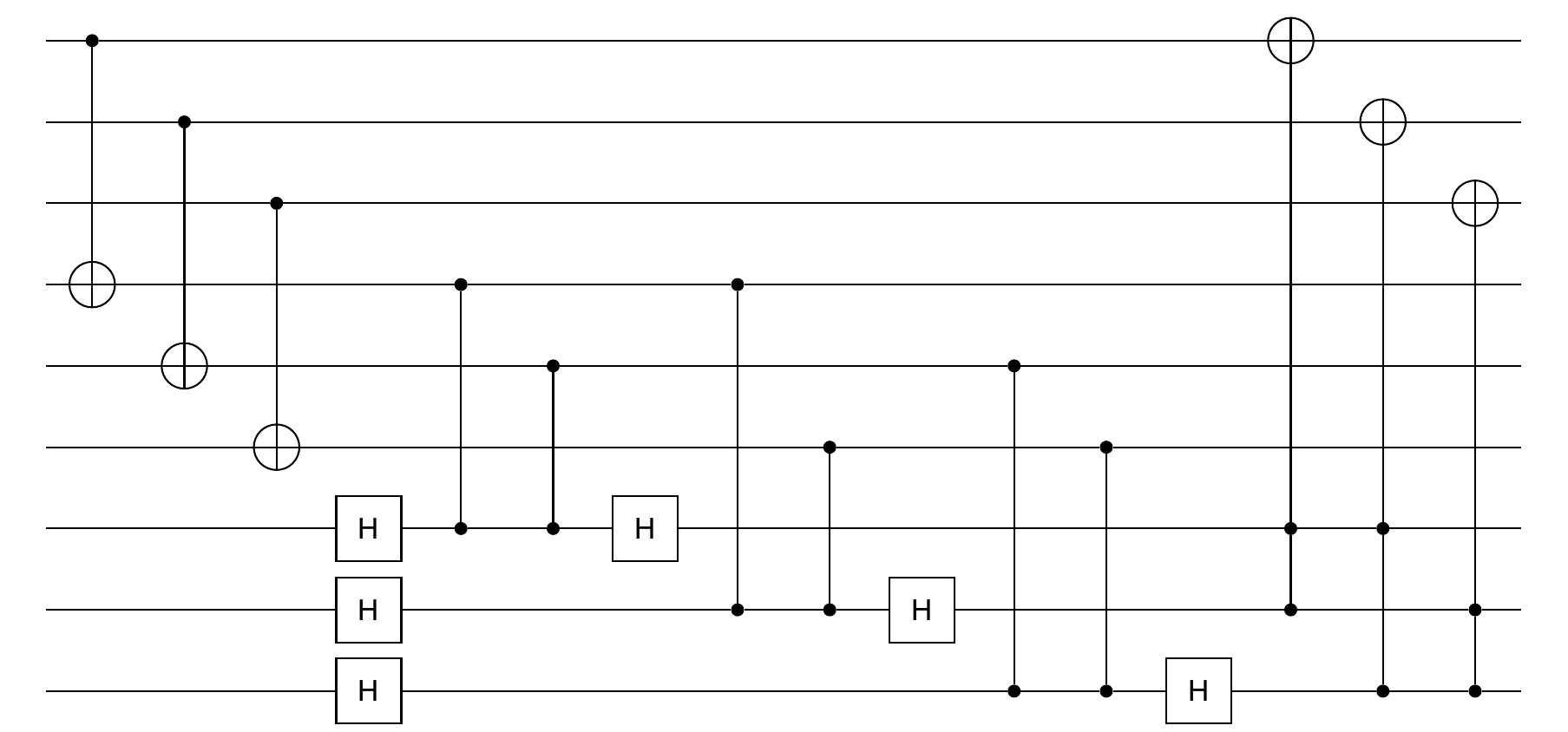}};

    \draw  [dash pattern={on 4.5pt off 4.5pt}]  (5,66) -- (330,66) ;
    \draw  [dash pattern={on 4.5pt off 4.5pt}]  (5,113) -- (330,113) ;
    
    \draw (10,126) node [anchor=north west][inner sep=0.75pt]  [font=\footnotesize]  {$|0\rangle ^{\otimes 3}$};
    \draw (10,79) node [anchor=north west][inner sep=0.75pt]  [font=\footnotesize]  {$|\overline{0} \rangle _{b}$};
    \draw (10,30) node [anchor=north west][inner sep=0.75pt]  [font=\footnotesize]  {$|\overline{\psi } \rangle _{b}$};
    \draw (24,6) node [anchor=north west][inner sep=0.75pt]  [font=\footnotesize] [align=left] {Error Copy};
    \draw (130,165) node [anchor=north west][inner sep=0.75pt]  [font=\footnotesize] [align=left] {Syndrome Extraction};
    \draw (238,6) node [anchor=north west][inner sep=0.75pt]  [font=\footnotesize] [align=left] {Controlled Correction}; 
\end{tikzpicture}
\caption{MF-QEC circuit for bit-flip repetition code. The circuit includes Steane type syndrome extraction (yellow) and controlled error correction~\cite{steane1996error} (cyan). The syndromes of the stabilizers are extracted to the last 3 qubits. Finally, the errors are corrected by Toffoli gates.}\label{fig:qcbf}
\end{figure}

\usetikzlibrary{decorations.pathreplacing}
\begin{figure}
    \begin{tikzpicture}[x=0.75pt,y=0.75pt,yscale=-1,xscale=1]
        
        \draw (255,100) node  {\includegraphics[width=0.62\columnwidth]{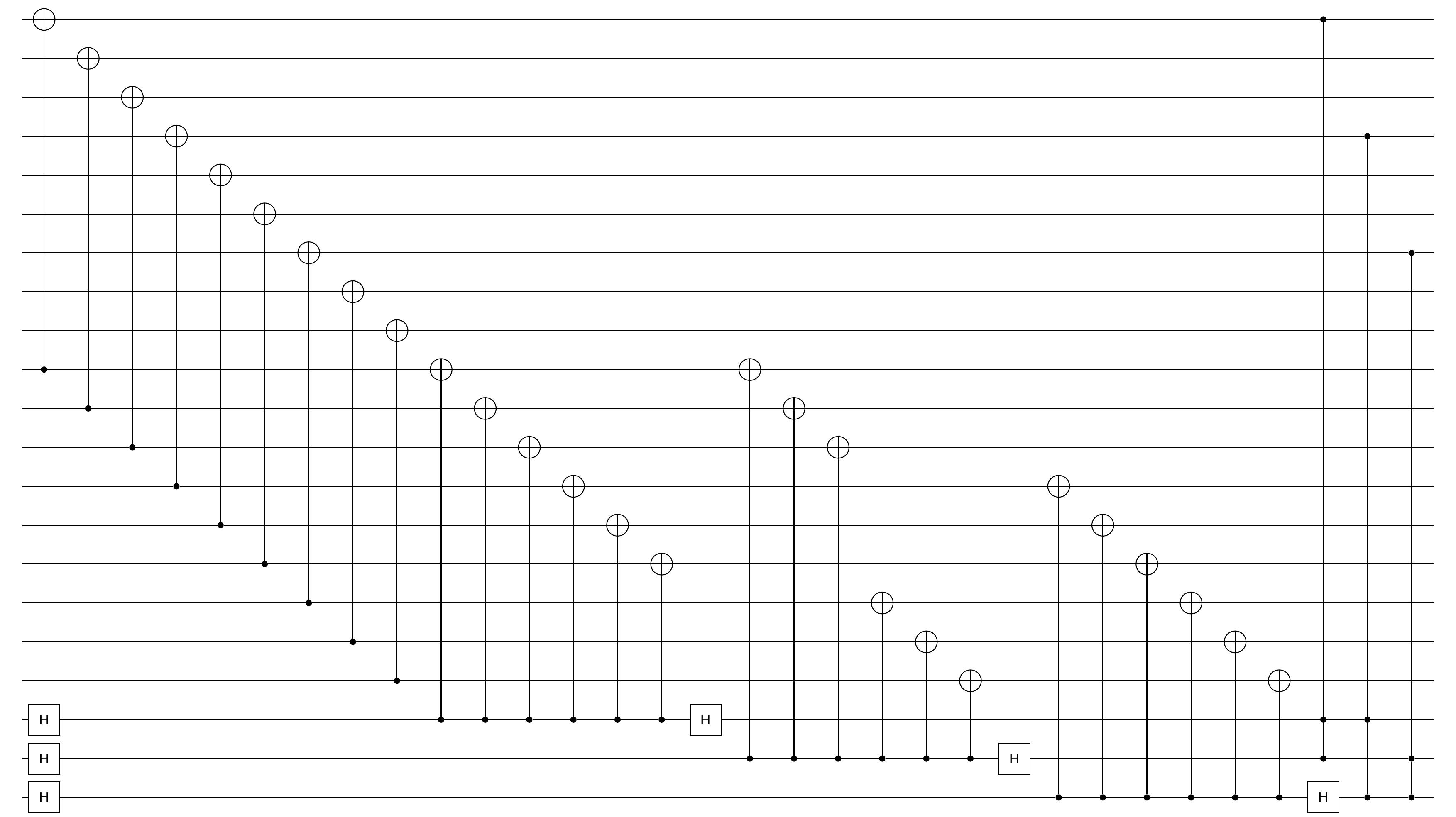}};
        
        \draw [decorate,decoration={brace,mirror,amplitude=2pt},xshift=35pt,yshift=-11pt](0.5,0.5) -- (0.5,95.0) node [black,midway,xshift=-0.6cm] {$|\overline{\psi} \rangle $};
        
        \draw [decorate,decoration={brace,mirror,amplitude=2pt},xshift=35pt,yshift=65pt](0.5,0.5) -- (0.5,95.0) node [black,midway,xshift=-0.6cm] {$|\overline{+} \rangle $};
        
        \draw [decorate,decoration={brace,mirror,amplitude=2pt},xshift=35pt,yshift=141pt](0.5,0.5) -- (0.5,28.0) node [black,midway,xshift=-0.6cm] {$|0\rangle ^{\otimes 3}$};
        
        \end{tikzpicture}
        \caption{The circuit for extracting syndromes of $X$ stabilizers and then correcting the corresponding errors. The syndromes are extracted to the last three qubits.}\label{fig:qcshorx} 
    \end{figure}

\begin{figure}
\begin{tikzpicture}[x=0.75pt,y=0.75pt,yscale=-1,xscale=1]

\draw (350,100) node  {\includegraphics[width=0.9\columnwidth]{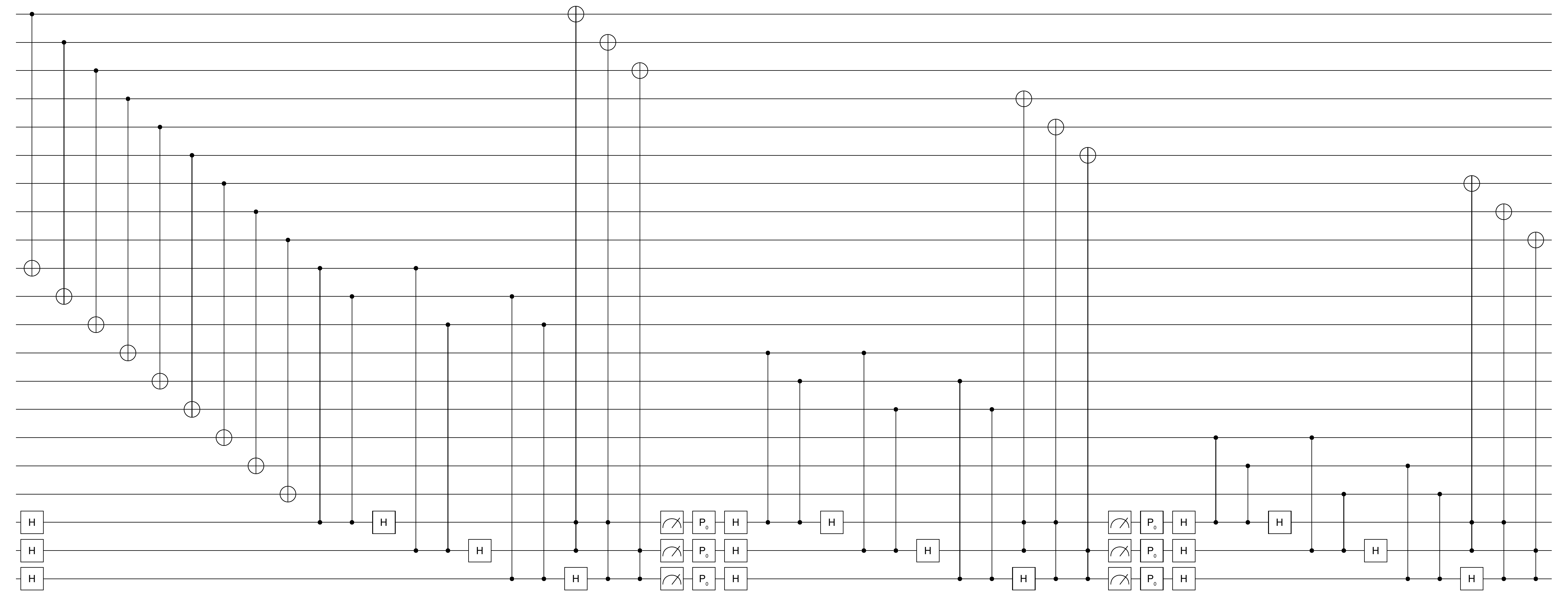}};

\draw [decorate,decoration={brace,mirror,amplitude=2pt},xshift=35pt,yshift=-10pt](0.5,0.5) -- (0.5,95.0) node [black,midway,xshift=-0.6cm] {$|\overline{\psi} \rangle $};

\draw [decorate,decoration={brace,mirror,amplitude=2pt},xshift=35pt,yshift=64pt](0.5,0.5) -- (0.5,95.0) node [black,midway,xshift=-0.6cm] {$|\overline{+} \rangle $};

\draw [decorate,decoration={brace,mirror,amplitude=2pt},xshift=35pt,yshift=139pt](0.5,0.5) -- (0.5,28.0) node [black,midway,xshift=-0.6cm] {$|0\rangle ^{\otimes 3}$};

\end{tikzpicture}
\caption{The circuit for extracting syndromes of $Z$ stabilizers and then correcting the corresponding errors. The syndromes are extracted to the last three qubits.}\label{fig:qcshorz} 
\end{figure}

\end{document}